\documentclass[12pt]{article}
\pdfoutput=1
\usepackage{amsmath, amssymb, amsfonts, amsthm}
\usepackage{graphicx,psfrag,epsf}
\usepackage{enumerate}
\usepackage{natbib}
\usepackage{url} 
\usepackage{xr}
\usepackage{xspace}
\usepackage{booktabs}
\usepackage{mathtools}
\usepackage{tabularx} 
\usepackage{hyperref}
\usepackage{multirow}
\usepackage{algorithm}
\usepackage{algorithmic}
\usepackage{subcaption}
\usepackage{color}
\usepackage{booktabs}
\usepackage{bm, stmaryrd}
\usepackage{comment}
\usepackage{verbatim}

\newcommand{\bR}{\mathbb{R}} 
\newcommand{\tr}{\text{trace}}
\newcommand{\that}{\hat{\theta}}
\newcommand{\tbar}{\bar{\theta}}

\newcommand{\bas}[1]{\begin{align*}#1\end{align*}}
\newcommand{\ba}[1]{\begin{align}#1\end{align}}
\newcommand{\cL}{\mathcal{L}}
\newcommand{\cN}{\mathcal{N}}
\newcommand{\bfeps}{\boldsymbol{\epsilon}}
\newcommand{\ttvec}{\text{vec}}
\newcommand{\beq}[1]{\begin{equation}#1\end{equation}}
\newcommand{\bsplt}[1]{\begin{split}#1\end{split}}

\newtheorem{definition}{Definition}
\newtheorem{theorem}{Theorem}
\newtheorem{corollary}{Corollary}

\newtheorem{lemma}{Lemma}

\newtheorem{remark}{Remark}

\newcommand{\sign}{\text{sign}}

\newcommand{\ttspan}{\text{span}}

\title{\bf \Large Total Variation Regularized Tensor-on-scalar Regression}

 \author{
Ying Liu\footnote{These two authors contributed equally to the work.}   \footnote{Correspond to yiliu{\it @}mcw.edu}\\ \small Division of Biostatistics, Medical College of Wisconsin, Milwaukee, WI \\
 \and  Bowei Yan\footnotemark[1] \\ \small Department of Statistics, University of Texas, Austin, TX\\
 \and Kathleen Merikangas\\\small Genetic Epidemiology Research Branch,\\\small  National Institute of Mental Health, Bethesda, MD \\
\and Haochang Shou\\\small Department of Biostatistics, Epidemiology and Informatics,\\ \small University of Pennsylvania, Philadelphia, PA 
 }
\begin{document}
\def\spacingset#1{\renewcommand{\baselinestretch}%
{#1}\small\normalsize} \spacingset{1}

\date{}
\maketitle
\begin{abstract}
In this paper, we propose {\it Total Variation Regularized Tensor-on-scalar Regression} (TVTR), a novel method for estimating the association between a tensor outcome (a one dimensional or multidimensional array) and scalar predictors. While the statistical developments proposed here were motivated by the brain mapping and activity tracking, the methodology is designed and presented in generality and is applicable to many other areas of scientific research.  The estimator is the solution of a penalized regression problem where the objective is the sum of square error plus a total variation (TV) regularization on the predicted mean across all subjects.  We propose an algorithm for the parameter estimation, which is efficient and  scalable in distributed computing platform. Proof of the algorithm convergence is provided and the statistical consistency of the estimator is presented via an oracle inequality.  We presented 1D and 2D simulation results, and demonstrate that TVTR outperforms existing methods in most cases. We also demonstrate the general applicability of the method by two real data examples including the analysis of the 1D accelerometry subsample of a large community-based study for mood disorders and the analysis of the 3D MRI data from the attention deficient/hyperactive deficient (ADHD) 200 consortium. 


\end{abstract}

%
\noindent%
\newpage
\spacingset{1.45}

\section{Introduction}
\label{sec:intro}

High-dimensional data continue to proliferate in modern biomedical research where many of outcomes are one-dimensional or multidimensional arrays (tensors) with complex spatial and temporal structures. For example, wearable computing sensors such as accelerometers and smart phones, have been increasingly used in epidemiological cohort studies such as National Health and Nutrition Examination Survey (NHANES) and UK Biobank \cite{naslund2015feasibility,pantelopoulos2010survey,patel2015wearable}. Instead of collecting subjects' physical activity levels through self-reported survey questionnaire that are scarcely sampled and subject to reporting biases, these devices have allowed us to observe physical activity intensities continuously in high frequencies over an extended period and have shown to be a more powerful way to track activity markers and circadian rhythm patterns that link with health-related behaviors and clinical outcomes such as mood, obesity and aging. Meanwhile in neuroscience, sustained growth in computing power and the increasing affordability to collect high-throughput medical images has made various modalities of brain imaging such as MRI, fMRI and PET available for researchers to search for promising brain markers that are associated with or predict the patient-specific diagnosis and prognosis. To untangle such associations, appropriate statistical methods that acknowledge the multidimensional structures of the outcome data and correlations induced by time or spatial continuity are needed. In particular, we were motivated by two specific studies that collect 1D curve and 3D array as outcome measures. The first example came from the accelerometry subsample of a large community-based study for mood disorders, the National Institute of Mental Health Family Study of Affective Spectrum Disorder \citep{Merikangas2014}. The study consist minute-by-minute physical activity intensity measures (i.e., activity counts) from a set of about 350 participants using the Philips Respronics 24 hours per day over 2 weeks. Despite that the observed activity data are highly noisy as shown in Figure \ref{raw}, we hypothesize that the underlying activity status (inactive, sedentary, moderate and vigorous) transit smoothly from one to another. The second study was from attention deficient/hyperactive deficient (ADHD) 200 consortium as part of the Human Connectome Initiative. The study included both structural and BOLD functional magnetic resonance imaging (MRI) scans of 776 subjects, including 491 typically developing controls (TDs) and 285 ADHD children from multiple centers. We were interested in investigating regions where gray matter concentration differs between TD and AD children by taking the voxel-based morphometry images as outcome in the model. In addition to the low signal-to-noise ratio, the images were also formulated in a three-dimensional spatially coherent space where gray matter exists. 

As previously noted, the challenges of developing a general and robust analytic framework for tensor response data are three-fold: the high dimensionality of the response, the low signal-to-noise ratio, and the complex dependency structures. 
While there exist an enormous body of literature tackling regression with high- or ultrahigh-dimensional data, they often deal with scalar outcomes and high-dimensional data as predictors through various regularizations and sparsity assumptions on the associations with the outcome variable. There have been relatively few works on regression with multivariate response, and even fewer works on regression with tensor responses. One line of research treats the multivariate observations as random realizations of functional objects and conduct function-on-scalar regression (FoSR) \citep{reiss2010fast,goldsmith2016assessing,scheipl2015functional}.
FoSR assumes both the outcomes and the regression coefficients are smooth over time or space, and can be represented in the space expanded by a set of pre-specified basis functions. When the outcome is a multi-dimensional array, one could reform the data into a vector and assumes smoothness within local time intervals or neighborhood of voxels. Such assumptions might violate the spatial continuity embedded in the original data that is beyond Euclidean distance.
Another line of research is reduced-rank multivariate regression~\citep{chen2012sparse,chen2013reduced,chen2012reduced,peng2010regularized}. Reduced-rank regression take advantage of inter-relationship within the response variables to improve prediction accuracy. 
Our method, unlike the aforementioned two categories, model the response data as a tensor object, where the vector response can be viewed as a special case of a 1D tensor. 
To the best of our knowledge, few papers have tackled tensor response with dimensional flexibility except the envelope methods~\citep{li2017parsimonious, Cook2010}.
Envelope method assumes a generalized sparsity in the data, and use BIC for selecting the latent dimension. However, the envelope methods \citep{li2017parsimonious} aims at multi-dimensional tensor outcomes not including 1D cases.

It is worth pointing out that our questions of interest are different from a class of tensor/functional regression models that quantify the relationship between a scalar outcome
and a functional/tensor regressor \citep{goldsmith2011penalized,zhou2013tensor,goldsmith2014smooth,wang2014regularized}. By contrast, our proposal reverses the role by treating the tensor object as response and the vector of covariates as predictors. These two kinds of problems have  different interpretations and applications. The former problem where the tensor object are predictors, focuses on assessing the effect on a clinical outcome as the biological signals (e.g. image) changes, and hence could be used for disease classification and prognosis. The latter where the tensor object is the response variable in regression, aims to identify patterns in the observed data and discover novel markers that are indicative of the disease status and demographics. Hence in this work, we focuses on the interpretability of the identified patterns and scientific conclusions are drawn from the insights gained by the estimated multi-dimensional patterns . 

Secondly, biomedical images or sensor data often suffer from low signal-to-noise ratio and it is common practice to denoise the observed signals before taking it into further processing. One of the most widely used techniques in signal and image processing is Total Variation (TV) denoising. 
It recovers blurred signals or images by penalizing the average intensity difference of adjacent records/pixels using $L_1$ norm. \cite{steidl2006splines} generalized the approach to penalize higher-order gradient difference and pointed out a connection between the dual formulation of TV denoising and spline smoothing of images. TV denoising can be viewed as a support vector regression problem in the discrete counterpart of the Sobolev space. 
Efficient algorithms have been proposed to implement the approach (see e.g. \cite{beck2009fast, condat2013direct,padilla2016dfs}). 
Total variation penalty has also been used in the statistics community to handle regression problems where the parameter has certain smooth structure. For example, fused lasso~\citep{tibshirani2005} was proposed by adding penalty between adjacent sites on genomes when assessing the association between phenotypes and genotypes.

Inspired by the idea of TV de-noising, we propose a regularization term in the regression problem to encourage smoothness in the outcome space. It is important to differentiate it from the fused lasso problem, which assumes the regression coefficients are smooth. Regularization directly imposed on the outcome image and tensor space simultaneously de-noises the complex structured input data in the model fitting step without either data pre-smoothing.
We will see in the experiments that this simultaneous treatment is crucial in preserving weak but consistent signals, that could be blurred out by off-the-shell de-noising algorithms. 

Thirdly, the proposed method can also easily handle complex dependency structure and preserve the inherited adjacency structure of the voxels or time points from the observed data. The smoothness constraint can go beyond adjacent voxels and be customized by any adjacency structures.  This could include, for example, the periodicity in longitudinal data such as seasonal effects in physical activity or temporal brain activation during a task fMRI scan. 

The proposed method is mainly designed for continuous outcomes and does not require explicit assumptions on the noise distributions. For convenience of theoretical analysis, we present an oracle inequality of the estimator, and show that it is statistical consistent when the noise is i.i.d. Gaussian. To solve the aforementioned regularized regression problem, we propose an Alternating Direction Method of Multipliers (ADMM) algorithm, which, our analysis proves, is globally convergent.
In both simulation studies and real data analysis, we compare our method with related work, and the proposed method exhibit strong performance especially when the sample size is small. Statistical inference on the estimated regression coefficients could be achieved using Bootstrap resampling that help the interpretability of the model in the real data analysis.

The rest of the paper is organized as follows. Section \ref{sec:method} describes the proposed method and introduces the model with some notations. In Section \ref{sec:algorithm}, we propose a fast and convergent algorithm to achieve the estimation.  We develop the consistency of the estimator via an oracle inequality in Section \ref{sec:consistency}. We demonstrate the performance in 1D and 2D settings by simulation follows in Section \ref{sec:simu}. We present the data analysis of an activity data of our proposed methods and compared to some other existing methods the result in \ref{sec:real}. We conclude the paper with some discussion in Section~\ref{sec:discussion}. Detailed proofs are outlined in the web appendices.

\section{Preparations}
\label{sec:method}

\subsection{Notations and Operations}
We begin with the introduction for the notations used in the paper. Throughout the paper we use lower-case and upper-case letters to represent vectors and matrices. Multidimensional array $A \in \bR ^{r_1 \times \dots \times r_m}$ is called an m{\it th}-order-tensor. 
 We will use $\|\cdot\|_F$ for Frobenius norm of a tensor,
$\|A\|^2_F= \sum_{i_1,\dots, i_m}A_{i_1,\dots, i_m}^2$. We denote the $\ell_1$ norm of a  as the $\ell_1$ norm of its vectorization: $\|A\|_{\ell_1}= \|\ttvec(A)\|_{\ell_1}=\sum_{i_1,\dots,i_m}|A_{ij}|$. $\otimes$ represents the Kronecker product between two tensors. The $\ttvec(\cdot)$ operator denote the vectorization operation that stack the entries of a tensor into column vector, so that an entry $a_{i_1,\dots,i_m}$ of $A$ maps to the $j$th entry of $\ttvec(\cdot)$, where $j=1+\sum_{k=1}^m(i_k-1)\Pi_{k'=1}^{k-1}r_k'$. and $\text{mat}(\cdot)$  denote the  the inverse operator. For an integer $n$, denote $[n]=\{1,2,\dots,n\}$.

\subsection{TV de-noising}
Next we review the total variation denoising objective where our method inherited the $L_1$ penalty for smoothing. The TV-denoising method for a 1D case,
$$\|Y-\theta\|^2 + \lambda\sum_{i=1}^{M-1} |\theta_{i+1}-\theta_{i}\|_{\ell_1},$$
where $Y\in \bR^n$ is a response vector, $\theta =(\theta_1,\dots,\theta_n)\in \bR^n = E(Y)$ is the mean of vector of $Y$, where we often assume $\theta_i=\theta_{i+1}$. 

In a m-dimensional setting, $Y$ is an $m^{\rm th}$ order tensor, $$\|Y-\Theta\|^2 + \lambda\|D \ttvec(\Theta)\|_{\ell_1},$$ where $D$ is an incidence matrix that we will introduce in the next session.

\section{Model and Inference}
\label{sec:model}
\subsection{Regression Model}
Without loss of generality, for a more intuitive presentation, we now introduce our method for 2D tensor outcomes. 
Suppose we observe the covariates and 2D tensor outcomes (for example 2D images) of size $m_1\times m_2$ from $n$ subjects $(X_i,\bar{Y}_i)_{i=1}^n$, where $X_i\in \bR^{p}, \bar{Y}_i\in \bR^{m1\times m2}$, $M=m_1 m_2$ is the total number of outcome entries and $Y_i=\ttvec({\bar{Y}_i}) \in \bR^{M}$ 
Consider the following model 
\bas{
Y_{i} =& X_i^T\Gamma + \epsilon_{i} =\sum_{t \in [p]} X_{it}\Gamma_{t\cdot}+\epsilon_i, 
}
where $\Gamma\in \bR^{p\times M}$ is the coefficient matrix of interest, and each row $\Gamma_{t \bullet}\in \bR^{M}$ is an image, representing the coefficient map corresponds to the $t${\it th} feature. 
While our theory is derived under the case of i.i.d. Gaussian noise, the method can be applied to other noise structure with real-value responses. 
Let $X\in \bR^{n\times p}$ be the design matrix and $Y\in \bR^{n\times M}$ be the matrix of outcomes.
Now we can stack all the observations and reformulate the problem into the following matrix form:
\ba{
\ttvec(Y^T) = \ttvec((X\Gamma)^T)+\bfeps, 
\label{model}
}
where 
$\ttvec(Y^T)=\left(\begin{array}{c}
Y_{1}\\Y_{2}\\\vdots\\Y_{n}\\
\end{array}\right) \in \mathbb{R}^{nM} $ is a long vector. 

Define the hat matrix as $H:=X(X^TX)^{-1}X^T$, which is consistent with that used in the classical regression analysis. We have the
projection matrix to $\text{span}(X)^{\perp}$ as $I-H$. In the vectorization form, we can define the extended projection matrix to project measurements from each voxel to the linear space $\text{span}(X)^{\perp}$.
\bas{
I_{nM}-H_v = (I-H) \otimes I_M\in \mathbb{R}^{nM\times nM}.
}
Notice that  for any $\Gamma  \in \mathbb{R}^{p\times M}$,
\bas{
(I-H_v) \ttvec((X\Gamma)^T) =& ((I-H) \otimes I_M )\ttvec(\Gamma^TX^T)= \ttvec(I_M\Gamma^TX^T (I-H)^T)=\bold{0}
}
Now let us introduce the smoothing graph in this problem. In a 2D or higher dimensional image outcome, it is usually the intrinsic nature that the adjacent voxels are similar, except for a small number of edges.  This can be summarized by defining a graph $G$, whose nodes are all the voxels in the image, and edges are the pairs of voxels which should have values that are close. For example, the grid graph is the most commonly used, where each voxel is connected to the four voxels adjacent to it in the image. The graph $G$ can also be chosen to reflect some sophisticated affinity relationship.

Let $D$ be the incidence matrix for the graph $G$ whose edges represent the smoothing affinity. To be specific, let $n$ and $m$ be the number of vertices's and edges respectively, $D\in \bR^{M \times m}$ where 
\[D_{i,j} = \left\{\begin{aligned} -1 & \mbox{ if the edge $e_j$ leaves vertex $v_i$;}\\ 1 &  \mbox{ if it enters vertex $v_i$;}\\ 0 & \mbox{ otherwise}\end{aligned} \right. \]
Note the orientation of the edge does not matter, since it corresponds to a negation of $D$, which does not change the $\ell_1$ norm.
Define the extended incidence matrix $D_v=I_n\otimes D \in \bR^{nM\times nm}$, then
$ \| X\Gamma D\|_{\ell_1}=\sum_{i=1}^n \|X_i\Gamma D\|_{\ell_1}$.
We propose the Total Variation Regularized Tensor-on-scalar Regression (TVTR) by minimizing the following loss.
\begin{equation}
\min_\Gamma \frac{1}{2}\| Y-X\Gamma\|_F^2 +\lambda \| X\Gamma D\|_{\ell_1}.
\label{eq:obj}
\tag{P}
\end{equation}
where $\lambda$ is a tuning parameter controlling the tradeoffs between the linear correlation and the smoothness of the fitted image.
It has been observed in \cite{steidl2006splines} that total variation regularization bears strong relationship with functional analysis with splines. To be more specific, it is shown in a strictly discrete setting that thin plate splines \citep{duchon1977splines} of degree $m-1$ solve also a minimization problem with quadratic data term and $m$-th order TV regularization term.



\subsection{Algorithm}
\label{sec:algorithm}
Our goal is to estimate the multi-dimensional coefficients $\Gamma$ in the objective function~\eqref{eq:obj}. Although it is a convex function, the non-smoothness of the total variation regularization makes traditional algorithms suffer from slow convergence. 
In many applications, the optimization problem involves
a large number of variables, and cannot be efficiently
handled by generic optimization solvers. 

Instead, we propose to use the Alternating Directional Method of Multipliers (ADMM).
ADMM was originally introduced by~\cite{glowinski1975approximation,gabay1976dual} and has been frequently used in distributed settings (see~\cite{boyd2011distributed} for an overview). The idea of using ADMM to solve total variation regularized problems were first derived by \citet{wahlberg2012admm}, which introduced an efficient and scalable optimization algorithm for problems satisfying  a {\it TV-separability condition}.  The condition requires the cost function is the sum of two terms, one that is separable in the variable blocks, and a second that is separable in the difference between consecutive variable blocks. In each iteration of our method, the first step involves separately optimizing over each variable block, which can be carried out in parallel.
The second step is not separable in the variables, but can be carried out by applying ADMM to $L_1$ mean filtering very efficiently. 
This algorithm has been successfully applied to fused lasso and total variation denoising problems and exhibited competitive performance. It is also very easy to extend to distributed computing setting. 


To make our problem satisfy the above {\it TV-separability condition}, we consider a reformulation of the original problem. Define $\theta=X\Gamma\in \bR^{n M}$, the problem is equivalent to,
\beq{
\bsplt{
\min_\theta \quad & \frac{1}{2}\| \ttvec(Y^T) -\theta\|^2_F +\lambda \| D_v^T\theta \|_{\ell_1};\\
\textrm{s.t.} \quad & (I-H_v)\theta=0.
}
\label{eq:cons_opt}\tag{CP}
}
The solution of \eqref{eq:obj} and that of \eqref{eq:cons_opt} have the following relationship.
\begin{lemma}
\label{lem:theta2gamma}
Let $\that$ be the optimal solution of \eqref{eq:cons_opt}, and $\hat{\Gamma}$ be the optimal solution of \eqref{eq:obj}. If rank($X)=p$, then 
$
\hat{\Gamma}=(X^TX)^{-1}X^T \text{mat}(\that)_{n\times M}
$.
\end{lemma}

Now to solve \eqref{eq:cons_opt}, we introduce two auxiliary variables and write the original problem in the following form.
\ba{
\min & \quad \frac{1}{2}\|y-\theta\|^2+\delta((I-H_v)\eta=0)+\lambda \|D_v^T\mu\|_{\ell_1} \nonumber \\
s.t. & \quad \theta=\mu, \theta=\eta \label{eq:augmented_admm}
}
where $\delta(\cdot)$ is the characteristic function, which takes 0 if the condition in the parenthesis is satisfied and infinity otherwise. 
The objective function is separable for $\theta, \eta$ and $\mu$, we solve it with the Alternating Direction Method of Multiplier (ADMM), which is summarized in Algorithm~\ref{alg:palm}.
\begin{algorithm}[h]
\caption{ADMM for Total Variation Regularized Tensor-on-scalar Regression}
\label{alg:alm}
\begin{algorithmic}[1]
\STATE {\bf Input}: Design matrix $X$, Images $Y$, 
tuning parameter $\rho$, error tolerance $tol$.
\STATE Initialize  $\theta^{(0)}, \mu^{(0)}, \eta^{(0)}$ randomly; $U^{(0)}=V^{(0)}=\bm{0}_{nM}$; $k=0$; 
\WHILE{not converge}
\STATE $\theta^{(k+1)}=(\textrm{vec} (Y)+\rho(\eta^{(k)}-U^{(k)}+\mu^{(k)}-V^{(k)}))/(2\rho+1)$;
\STATE $\eta^{(k+1)}= H_v(\theta^{(k+1)}+U^{(k)})$;
\STATE $\mu^{(k+1)} = \arg\min_\mu \lambda\|D_v^T\mu^{(k)}\|_{\ell_1}+\frac{\rho}{2} \|\theta^{(k+1)} +V^{(k)}-\mu^{(k)}\|_F^2$;
\STATE $U^{(k+1)} = U^{(k)} +\theta^{(k+1)} -\eta^{(k+1)}$;
\STATE $V^{(k+1)} = V^{(k)}+\theta^{(k+1)} - \mu^{(k+1)}$;
\STATE $k=k+1$;
\STATE converge {\bf if} $\max\{ \|\theta^{(k+1)}) - \theta^{(k)})\|_F/\|\theta^{(k)})\|_F, \|H_v\theta^{(k+1)}\|_F\} < tol$.
\ENDWHILE
\STATE {\bf Output}: $\Gamma=(X^TX)^{-1}X^T\text{mat}(\theta^{(k+1)})_{n\times M}$;
\end{algorithmic}
\label{alg:palm}
\end{algorithm}

In Algorithm~\ref{alg:palm}, line 4-6 are the primal variable updates and line 7-8 are the dual variable updates. It is common practice to set $\rho=1$. The update of $\theta$ in line 4 and projection step in line 5 both have analytical form and only consist of matrix multiplication. The computational bottleneck is the subproblem in line 6, which can be recognized as a graph fused lasso (GFL) problem that takes observation $\theta_{k+1}+V_k$, regularization parameter $\lambda/\rho$, and graph incidence matrix $D_v^T$.
GFL is a generalization of fused lasso. It penalizes the first differences of the signal across edges and can be efficiently solved with many off-the-shell GFL solvers. Below we highlight several solvers for GFL, and compare their efficiency under different assumptions.

When the graph is a chain graph, the GFL reduces to a 1D fused lasso problem, and can be solved in $O(n)$ time, by, for example, the ``taut string" algorithm derived by \cite{davies2001local} and the dynamic programming based algorithm by \cite{johnson2013dynamic}. These methods are fast in practice. 
For 2D cases, \cite{kolmogorov2016total} generalizes the dynamic programming idea to solve the fused lasso problem on a tree in $O(n\log n)$ time. \cite{barbero2014modular} uses operator splitting techniques like Douglas-Rachford splitting and extends fast 1D fused lasso optimizers to work over grid graphs. Over general graphs structure, numerous algorithms have been proposed in recent years: \cite{chambolle2009total} described a direct algorithm based on a reduction to parametric max flow programming; \cite{chambolle2011first} described what can be seen as a kind of preconditioned ADMM-style algorithm; \cite{kovac2011nonparametric} described an active set approach; \cite{landrieu2016cut} derived a method based on graph cuts. In our experiments, we use the one proposed in \cite{tansey2015fast} which leverages fast 1D fused lasso solvers in an ADMM decomposition over trails of the graph. 
One can choose the graph fused lasso solver that best suits the target graph structure.

Despite being a popular algorithm for many constrained optimization problems, the convergence of ADMM does not have a simple and unified answer. We present the convergence of Algorithm~\ref{alg:palm} in the following theorem. 
\begin{theorem}
Algorithm~\ref{alg:palm} converges linearly to the unique global optimum of problem~\eqref{eq:cons_opt}.
\label{thm:admm_converge}
\end{theorem}
To obtain the result, we first convert problem~\ref{eq:cons_opt} into a two-block problem; then apply recent result in general two-block problems~\cite{nishihara2015general} to complete the proof. The detailed proof is deferred to supplementary materials.

\subsection{Software}
We provide the open source python code for implementation of Algorithm~\ref{alg:palm} on \url{https://github.com/summeryingliu/imagereg}.

\section{Theoretical Results}
\label{sec:consistency}

In this section, we analyze the statistical property of the estimator given by \eqref{eq:obj}. We present the result via an oracle inequality on the prediction error of the regression. Before we state the theoretical result, we first introduce two quantities which are widely used in the analysis of sparse recovery problems.
\begin{definition}[Compatibility factor]
Let $D\in \bR^{M \times m}$ be an incidence matrix. The \it{compatibility factor} of $D$ for a set $\emptyset \subsetneq T\subset [m]$ is defined as
\bas{
\kappa_T := \inf_{\theta \in \bR^T} \frac{\sqrt{|T|}\|\theta\|}{\|(\theta D)_T\|_{\ell_1}};\quad \kappa=\inf_{T\subset [m]} \kappa_T
}
\vskip -0.05in
\label{def:compat_factor}
\end{definition}
Compatibility factor gets its name based on the idea that, on the subset of edges indicated by set $T$, we require the $\ell_1$-norm and the $\ell_2$-norm to be somehow compatible. Compared with other conditions used to derive sparsity oracle inequalities, such as restricted eigenvalue conditions or irrepresentable conditions, the compatibility factor greater than 0 is shown to be weaker by \cite{van2009conditions}. More discussion about the relationship between different conditions can be found in \cite{van2009conditions}. For graphs with bounded degree, it is shown that the compatibility condition is always satisfied.

\begin{lemma}[\citet{hutter2016optimal}]
Let $D$ be the incidence matrix of a graph $G$ with maximal degree $d$ and $\emptyset \ne T\subset E$. Then, $ \kappa_T \ge \frac{1}{2\min\{\sqrt{d},\sqrt{|T|}\}}$.
\label{prop:compat}
\end{lemma}

\begin{definition}[Inverse scaling factor]
The \it{inverse scaling factor} of an incidence matrix $D$ is defined as 
$
\rho(D): =\max_{j\in [m]} \|s_j\|$,
where $S=(D^T)^\dagger=[s_1^T,\cdots, s_m^T]^T$ is the pseudo inverse of $D^T$.
\label{def:inv_scaling_factor}
\end{definition}
By design of $D_v$, it is clear that $\rho(D_v)=\rho(D)$. 
Now we present the main result.
\begin{theorem}[Oracle Inequality for Projected TV Regression]
Under model \eqref{model} with $\epsilon_{i}\stackrel{i.i.d.}{\sim} \cN(0,\sigma^2I_M)$, define $\theta^*=\ttvec(\Gamma^{*T}X^T) $, $\that$ is the solution for \eqref{eq:cons_opt}. For any $\delta>0$, if $\lambda=\rho\sigma\sqrt{\log(mnM/\delta)}$, then with probability at least $1-\delta$,
\bas{
 \|\theta^*-\hat{\theta}\|_F^2 \le & \inf_{\bar{\theta}\in \bR^{n M}: H_v\bar{\theta}=\bar{\theta}} \left\{ \|\bar{\theta}-\theta^*\|^2+ 4\lambda \|(D_v^T \bar{\theta})_{T^c}\|_{\ell_1} \right\} \\
& \quad + 64\sigma^2\log\left( \frac{2enM}{\delta} \right)+ 8\rho^2\sigma^2\log\left(\frac{mnM}{\delta}\right)\kappa_T^{-2}|T|}
\label{th:consist}
\end{theorem}
The proof of Theorem~\ref{th:consist} is inspired by that in \cite{hutter2016optimal} and is deferred to Section~\ref{sec:proof_consistent} in supplementary materials.
\begin{remark}
The length of $\theta$ scales as $nM$. Hence we care about the mean recovery error $\frac{1}{nM}\| \theta^*-\hat{\theta} \|^2$, which converges at rate $O\left( \frac{\log(mnM)}{nM} \right)$. The upper bound exhibits the trade-off between two quantities: the number of changing point $|T|$ and the total variation of the ``smooth" part  $\|(D\theta)_{T^c}\|_{\ell_1}$. To be more specific, given the total variation of the entire data fixed, when the model is piecewise smooth but have drastic change at the changing point, $|T|$ will dominate $\|(D\theta)_{T^c}\|_{\ell_1}$; on the other hand, if the data has few changing point but fluctuate a lot in each piece, the total variation of $\theta$ in $T^c$ will be large compared to $|T|$.
\end{remark}

If the graph has bounded degree, by Lemma~\ref{prop:compat} we immediately have the following corollary.
\begin{corollary}
If the maximal degree of the penalty graph $G$ is $d$, $\lambda=\rho\sigma\sqrt{\log(mnM/\delta)}$, then with probability at least $1-\delta$,
\bas{
\|\theta^*-\hat{\theta}\|_F^2 \le&  \inf_{\bar{\theta}\in \bR^{n M}: H_v\bar{\theta}=\bar{\theta}} \left\{ \|\bar{\theta}-\theta^*\|_F^2+ 4\lambda \|(D_v^T\bar{\theta})_{T^c}\|_{\ell_1} \right\} \\
&\qquad + 64\sigma^2\log\left( \frac{2enM}{\delta} \right)+ \frac{2\rho^2\sigma^2\log\left(\frac{mnM}{\delta}\right)|T|}{\min\{d,|T|\}}
}
\end{corollary}

In particular, if the features has isotropic covariance matrix, then we will have the following bound on the parameter estimation.
\begin{corollary}
If  $\frac{1}{n}X^TX=I_p$, and $\lambda=\rho\sigma\sqrt{\log(mnM/\delta)}$, then with probability at least $1-\delta$,
\bas{
 \frac{1}{Mp}\|\hat{\Gamma}-\Gamma^*\|_F^2 \le &  \inf_{\bar{\Gamma}\in \bR^{p\times M}} \left( \frac{1}{Mp}\|\bar{\Gamma}-\hat{\Gamma}\|_F^2+ \frac{4\lambda}{ nMp} \|(X\bar{\Gamma}D)_{T^c}\|_{\ell_1} \right)  \\
&\quad + \frac{64\sigma^2\log\left( \frac{2enM}{\delta}\right) }{nMp }+ \frac{8\rho^2\sigma^2\log\left(\frac{mnM}{\delta}\right)|T| }{nMp\kappa_T^2}
}
\label{cor:gamma_oracle}
\end{corollary}

\begin{remark}
The condition in Corollary~\ref{cor:gamma_oracle} can be achieved by normalizing the input design matrix. When the covariance matrix for the features is not isotropic, the error term should be represented in Mahalanobis distance instead, that is, the error along different axis needs to be reweighted by the variance in that direction.  
\end{remark}

\section{Simulation Studies}
\label{sec:simu}

In this section, we present the simulation results on recovering signal $\Gamma$ with various methods. Since the goal is to accurately estimate $\Gamma$, we use the mean devision (square root of the mean squared error) of the coefficients as the performance metric, which is defined as $\frac{1}{\sqrt{Mp}}\|\hat{\Gamma}-\Gamma^*\|_F$. The standard deviations of the metric over replicates are reported in the parenthesis.

We conduct experiments on 1D and 2D synthetic data. For 1D data, we compare our method with the state-of-the-art function-on-scalar regression with B-spline and Fourier basis, which is implemented in R package \texttt{refund}; 
The setting and results are summarized in Section~\ref{sec:1dsimu}. For 2D data, the function-on-scalar method is applied on the vectorized image. We also compare with 
Parsimonious Tensor Response Regression~\citep{li2017parsimonious}, and some two-step procedures which we discuss in Section~\ref{sec:2dsimu}. 

\subsection{One Dimensional Simulation}
\label{sec:1dsimu}
We generate 2 simulation scenarios, each with 200 replications. For each scenario, we present results for sample sizes of 25, 50 and 100. 
In the first setting we generate model that satisfies the assumptions of function-on-scalar regression, i.e., the signals are sparse if expanded in Fourier basis. The second setting is generated not following this model assumption.

For the first setting, the predictors are generated as following: $(X_1, X_2)$ are drawn from categorical distribution that takes value $(1,0)$ with probability $1/4$, $(0,1)$ with probability $1/4$, and $(0,0)$ with probability $1/2$. $X_3$ is drawn from standard normal distribution. The true signals are from Fourier basis functions, with index $t$ ranges from $1$ to $200$, and $$Y=0.3 \sin (\pi t/100)+0.5 X_1 \cos (\pi t/100)-0.3 X_2 \sin (\pi t/50)+ 0.5 \cos(\pi t/25) + 2 \mathcal{N}(0,1).$$ We compared our proposed method with Functional-on-Scalar regression methods (FoSR) using B-spline and Fourier basis functions with $L_2$ penalty, where the tuning parameter is selected through cross validation from a grid of $(2^{-5},2^{-3},2^{-1},2,2^3,2^5)$. 
It is noticeable that scenario 1 is in favor of the FoSR methods since it is generated from the true model assumption with Fourier basis.  The Parsimonious Tensor Response Regression~\citep{li2017parsimonious} is proposed only for tensor regression problem and not intended for 1D function-on-scalar regression problem.

Table \ref{1Dsimu} presents the simulation results for scenario 1. 
When sample size is $25$, the proposed method shows a better performance than all the competitors in terms of smaller mean and standard devision. Our method performs similarly with the FoSR methods when sample size is 50. With large sample size 100, the FoSR method with the correct basis has the best performance. The Fourier basis function performs slightly better then b-spline basis, however the difference is not phenomenal. 

\begin{table}[ht]
\centering
\footnotesize
\caption{Mean Deviation of  Coefficients for 1D settings}
\begin{tabular}{l|p{1cm}|llll}
  \hline
& sample size	& TVTR 		&periodic\_TVTR			& FoSR\_bspline 		& FoSR\_fourier 	\\ 
  \hline
\multirow{3}{*}{Setting 1}  	&25 		&{\bf 0.045(0.008) }	& - &0.054(0.015) 	& 0.053(0.015) 		 \\ 
  					&50 		& 0.024(0.004) 		& - &0.024(0.006) 	& 0.024(0.006) 		 \\ 
					&100 	& 0.015(0.002) 		& - &0.012(0.003) 	&{\bf 0.011(0.003)} 	 \\ 
   \hline
 \multirow{3}{*}{Setting 2}	&25 		& 0.076(0.009) 		& {\bf 0.033(0.007)} 	& 0.081(0.015) & 0.078(0.015)\\ 
 					&50 		& 0.041(0.004) 		&{\bf 0.020(0.004)} 	& 0.051(0.006) & 0.049(0.006)  \\ 
					&100 	& 0.020(0.002) 		& {\bf 0.012(0.002)}	& 0.039(0.003) & 0.037(0.004) \\ 
   \hline
\end{tabular}
\label{1Dsimu}
\end{table}
The second simulation setting showed the advantage of our methods where the function on scalar regression model are misspecified. It also demonstrates the advantage of our proposed method in the flexibility for defining arbitrary adjacency structure with self defined incidence matrix.  
We generate piecewise constant signals, and in addition the signals for indexes 1 to 100 are identical with signals in 101:200. This setting is motivated by the scenario in time series observations, where one may have observations of the same pattern in repeated time periods.  For example, periodicity that correspond to day of the week or season of the year might be observed in subjects' levels of physical activity. Similar brain activation time series might also manifest during a task fMRI experiment when the stimulations were given repeatedly. We used this 1D simulation to demonstrate the performance of incorporating some prior knowledge compared with not using this knowledge. 

We consider two regimes in defining the edges. For vanilla TVTR, we only consider edges connecting adjacent observations: $E=\{(i,i+1), i=1,\dots,199\}$. For the periodic\_TVTR, we add in 100 more edges that connects the $i$th and $100+i$th observations, which encourages the periodicity of the signal.
Here the edge set is $E=\{(i,i+1), i=1,\dots,199\} \cup \{(i,i+100),i=1,\dots,100\}$. 
The true signal is generated by the following model:
\bas{
&Y_{ij} = \mathbb{I}(1\le j\le 20)+ \mathbb{I}(101\le j \le 120)+ 0.5 X_{1i} \left(\mathbb{I} (31\le j \le 70)+\mathbb{I}(131\le j \le170)\right) \\
& -X_{2i} \left(\mathbb{I}(71 \le j \le 80)+\mathbb{I}(171 \le j \le 180)\right) + X_{3i} \left(\mathbb{I}(61 \le j \le 100)+\mathbb{I}(161 \le j \le 200)\right) + 2\mathcal{N}(0,1).
}
The results are shown in Table \ref{1Dsimu}. 
We can see TVTR with or without the added edges outperforms the other methods. As sample size increases all the methods improves with smaller mean and SD for the mean deviation. By adding the additional edges, our proposed method shows better performance, and the advantage is more phenomenal with small sample size. The proposed method with added edges to incorporate prior knowledge and encourage similar patterns reduced the mean deviation to less than half for the one estimated  without added edges. 

\subsection{Two Dimensional Simulation}
\label{sec:2dsimu}
\begin{table}
\centering
\footnotesize
\caption{Mean Deviation of Coefficients for 2D settings}
\vskip 0.1in
\begin{tabular}{c|p{1cm}|lllll}
\hline
		& sample size	& TVTR		& Envelope 	& FoSR	& TV\_OLS	&OLS\_TV \\
		\hline
Setting 1	& 25			& {\bf 0.298 (0.038)} 	& 0.349(0.028)		& 0.357 (0.019) 		& 0.501 (0.044)	 	& 0.424 (0.048)\\
		& 50			& {\bf 0.217 (0.022)} 	& 0.241(0.013)		&0.339 (0.010)			&0.361 (0.028)	 	&0.313(0.033)\\
		& 100		& 0.200 (0.016)		& {\bf 0.169(0.007)}	&0.330 (0.004) 			&0.292 (0.017)	 	&0.230 (0.013)\\
		\hline
Setting 2	& 25			&{\bf 0.265 (0.008)} 	& 0.561(0.045)		&0.321 (0.003) 			& 0.275 (0.006) 	& 0.919 (0.087)\\
		& 50			& {\bf 0.192 (0.007)}	& 0.387(0.019)		&0.318 (0.001) 			& 0.266 (0.003) 	& 0.545 (0.033)\\
		& 100		& {\bf 0.126 (0.006)} 	& 0.267(0.010)		&0.318 (0.001) 			&0.259 (0.003)		& 0.277 (0.012)
\\\hline
\end{tabular}
\label{table2}
\end{table}
In this section, we present the experimental results for 2D synthetic data. All tuning parameters in the proposed method are chosen by 4-fold cross validation with a grid search in $[0.1, 0.25, 0.5, 1, 1.5, 2, 3]$. 
Some methods such as the one proposed in~\citet{reiss2010fast} fail to adapt to the scale we consider in this simulation due to the large memory cost. We demonstrate the performance of the function-on-scalar framework using penalized splines~\citep{goldsmith2016assessing}, with an efficient variational Bayes implementation (\texttt{bayes\_fosr} in R package \texttt{refund}). For envelope method proposed by~\citet{li2017parsimonious}, we use the matlab package provided by the authors.
We also consider two two-stage methods: TV-smooth before and after the regression, where the regression model is the voxel-wise ordinary linear regression. 

All methods used in the simulation are summarized as following.
\begin{itemize}
\item \texttt{TVTR}: Total Variation regularized Tensor-on-scalar Regression by Algorithm \ref{alg:palm}.
\item \texttt{Envelope}: Parsimonious Tensor Response Regression by \cite{li2017parsimonious};
\item  \texttt{FoSR}: A variational Bayes implementation for penalized splines.
\item \texttt{TV\_OLS} : Each image is denoised individually by total variation regularization, and the estimator is achieved by conducting a voxel-wise OLS on the denoised images.
\item  \texttt{OLS\_TV}: A TV denoised estimator from voxel-wise OLS regression.
\end{itemize}

\begin{figure}[t]
\small
\begin{tabular}{cc}
\hspace{-1em}\includegraphics[width=0.52\textwidth]{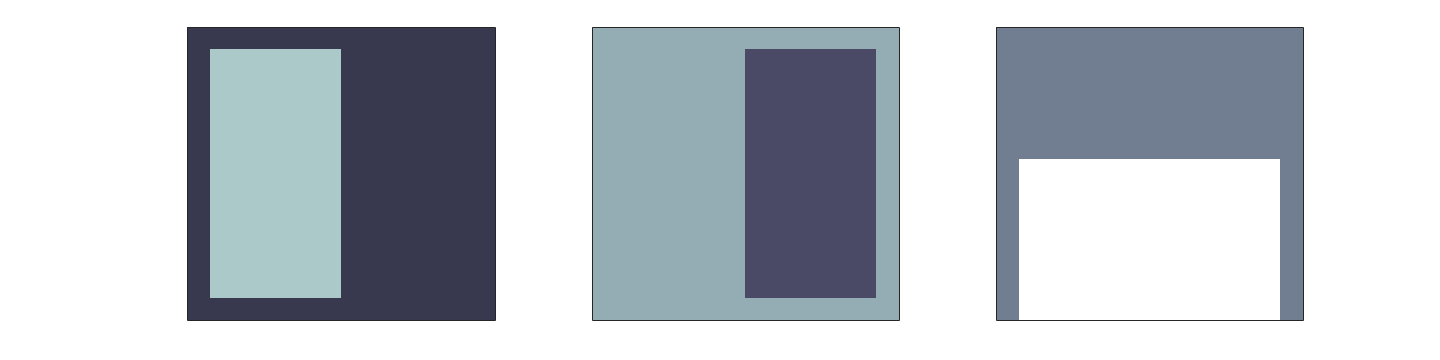} \hspace{-2em}	 	&\includegraphics[width=0.52\textwidth]{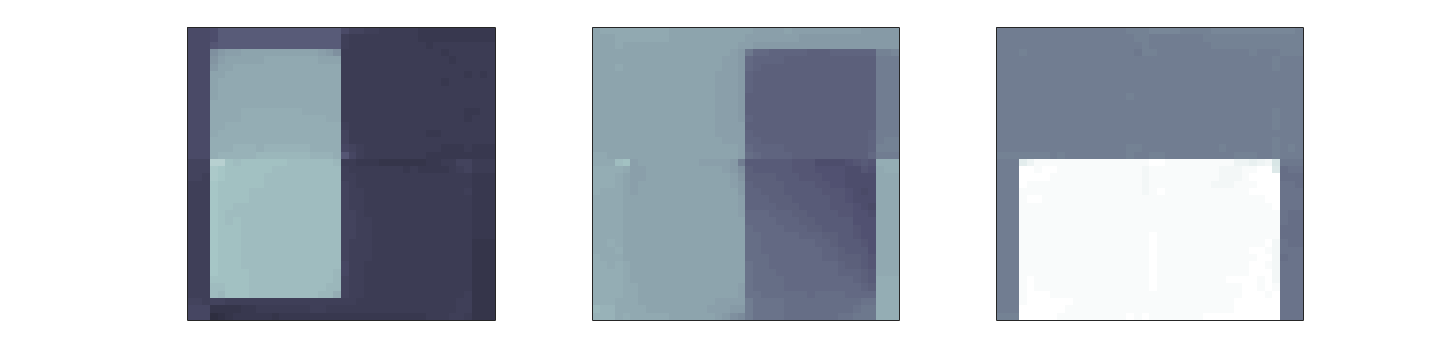}	\\
(a) Ground truth & (b) TVTR\\
\hspace{-1em}\includegraphics[width=0.52\textwidth]{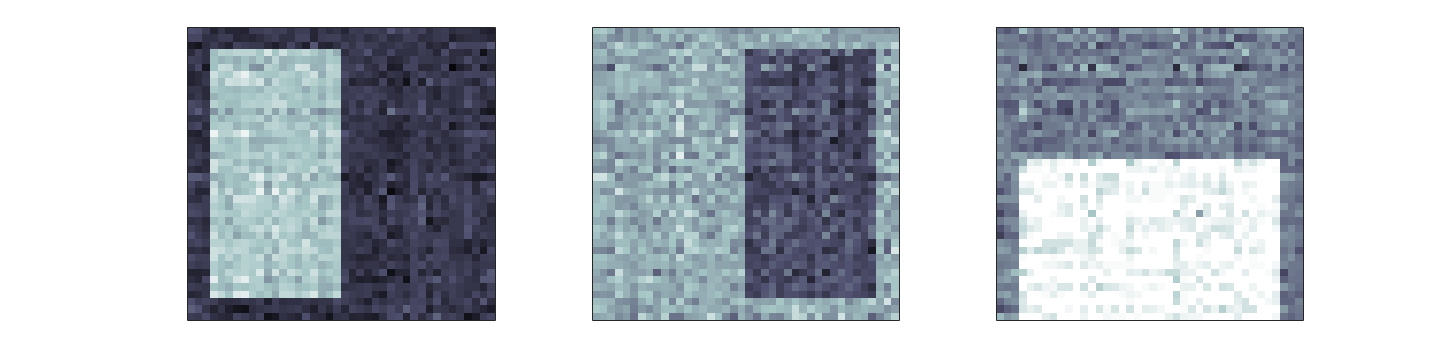} \hspace{-2em}	 	&\includegraphics[width=0.52\textwidth]{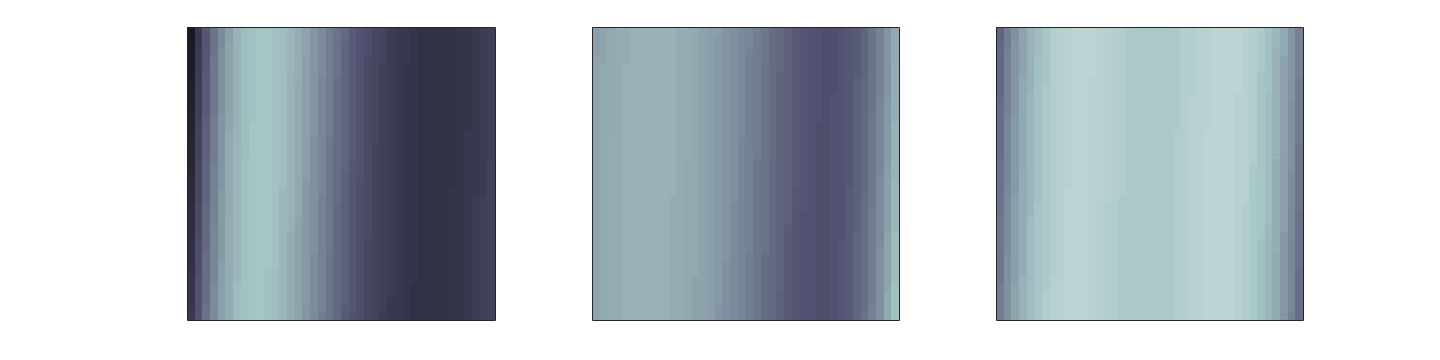}	\\
(c) Envelope & (d) FoSR\\
\hspace{-1em}\includegraphics[width=0.52\textwidth]{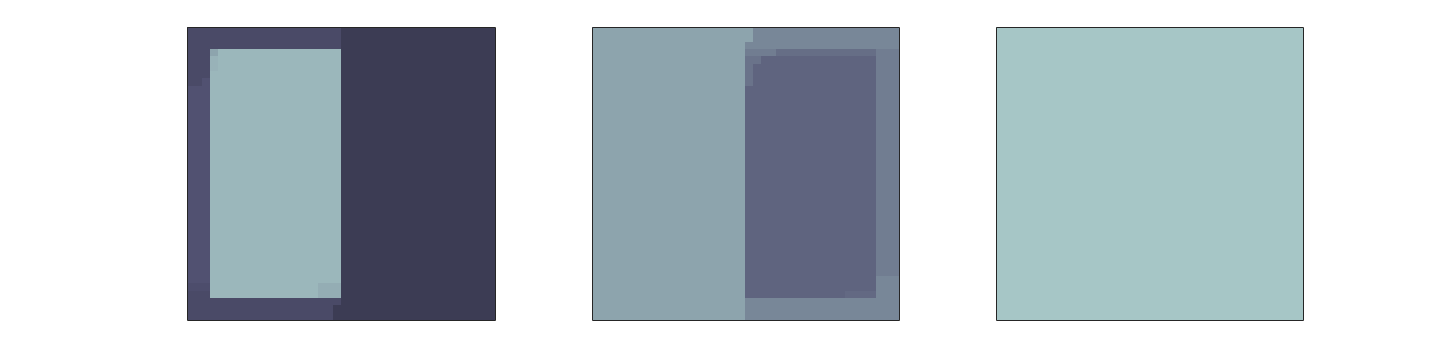} \hspace{-2em} 	&\includegraphics[width=0.52\textwidth]{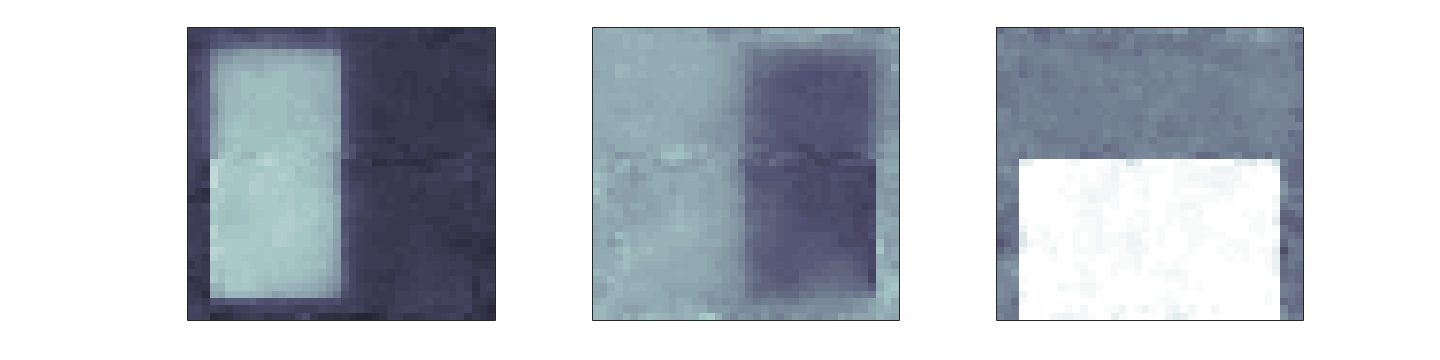}	\\
(e) Two step: OLS\_TV & (f) Two step: TV\_OLS
\end{tabular}
\vspace{1cm}
\caption{Coefficient Maps for Simulation Setting 1 in 2D.}
\label{coefexample1}
\end{figure}

\begin{figure}[t]
\small
\begin{tabular}{cc}
\includegraphics[width=0.44\textwidth]{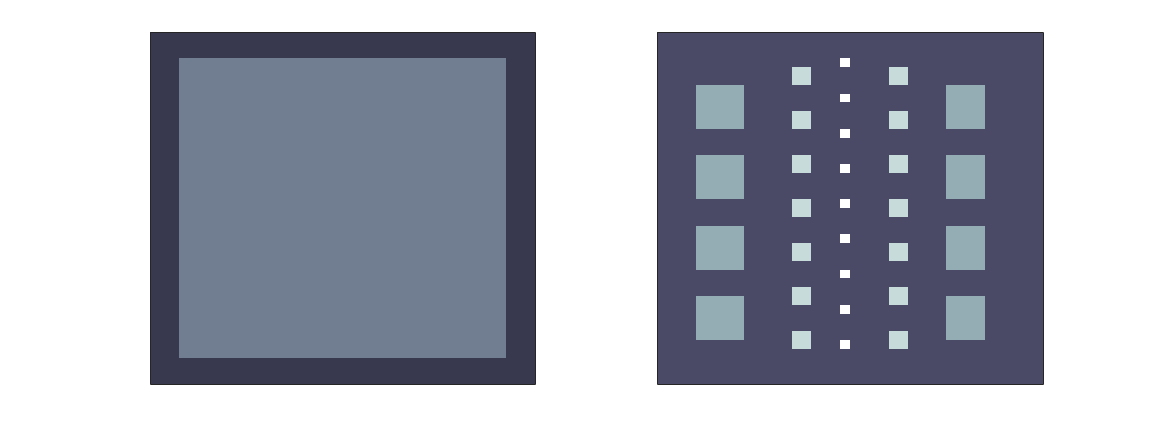}		&\includegraphics[width=0.44\textwidth]{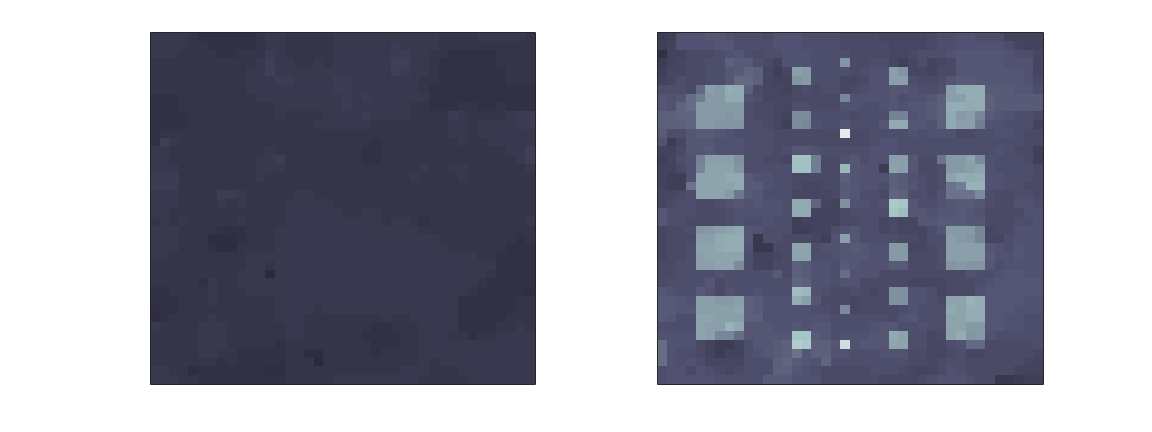}	\\
(a) Ground truth & (b) TVTR\\
\includegraphics[width=0.44\textwidth]{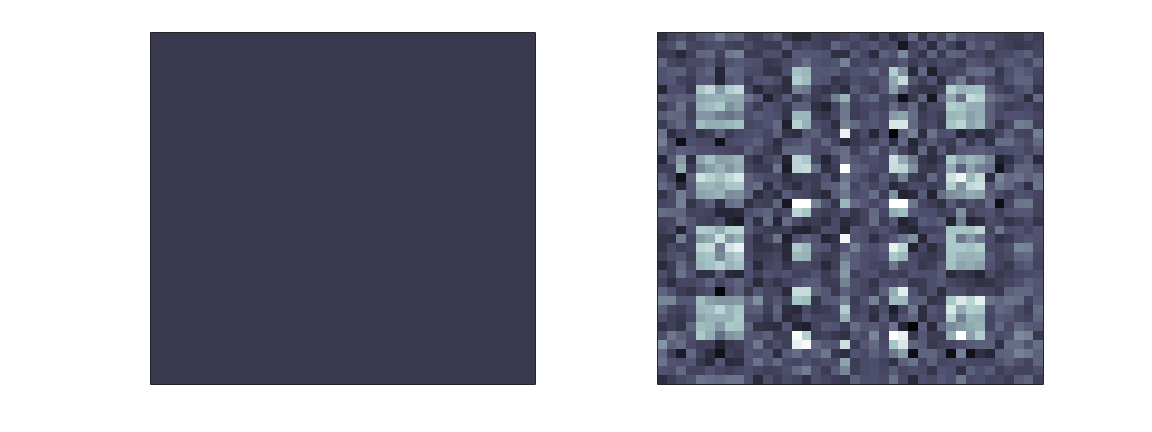}		&\includegraphics[width=0.44\textwidth]{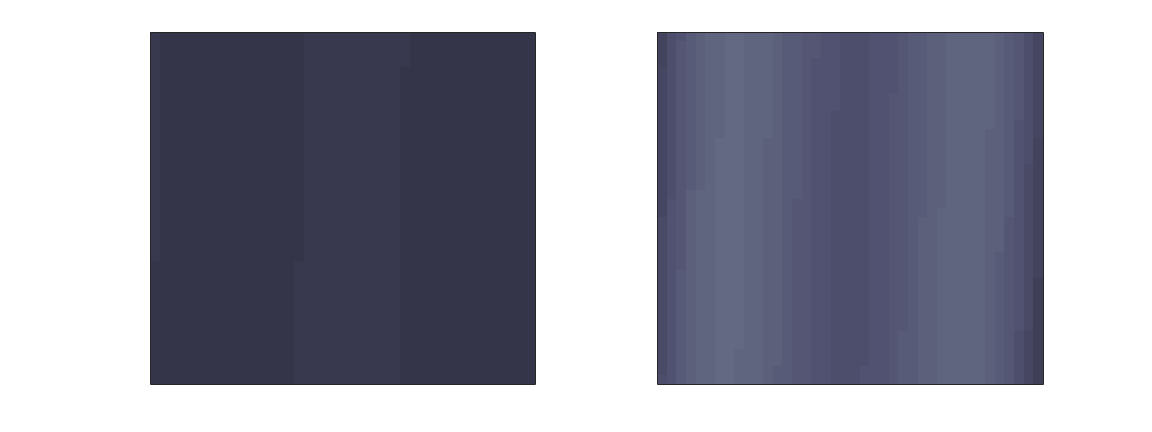}	\\
(c) Envelope & (d) FoSR\\
\includegraphics[width=0.44\textwidth]{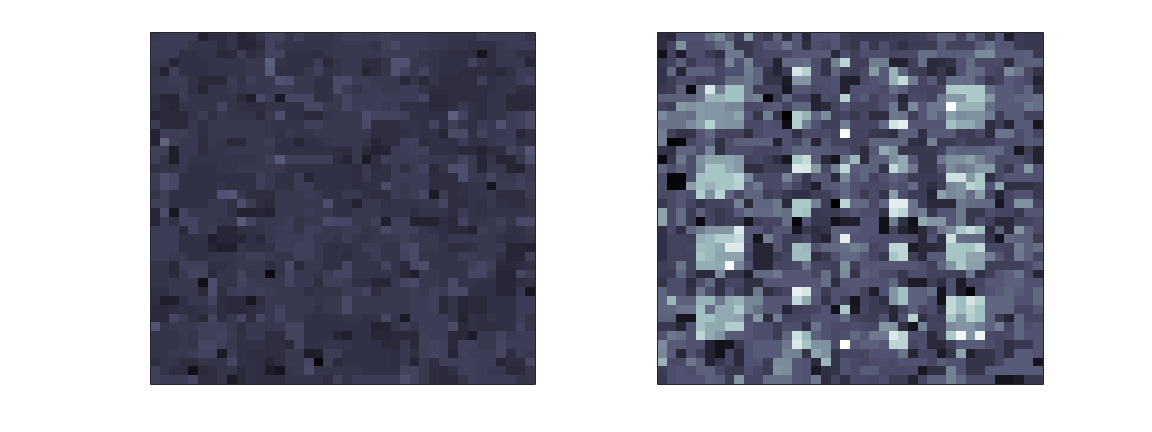}		&\includegraphics[width=0.44\textwidth]{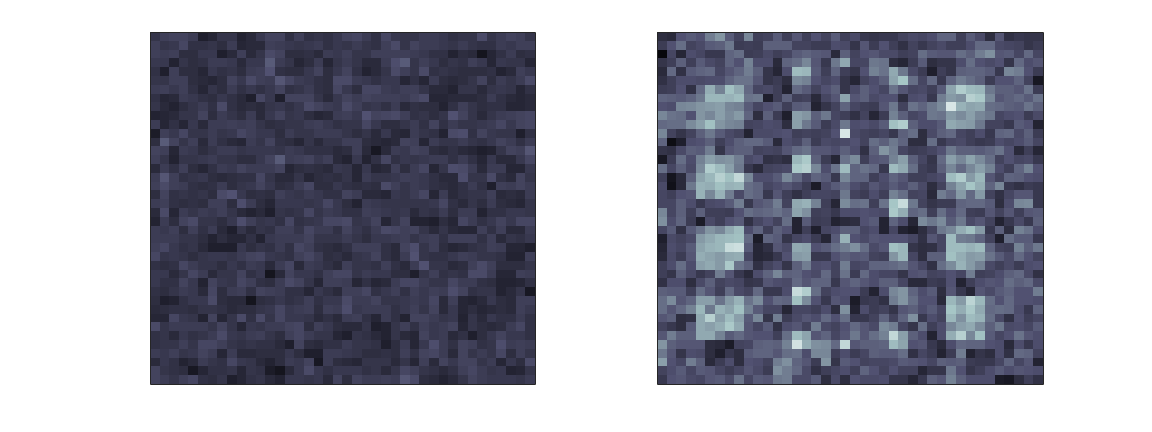}	\\
(e) Two step: OLS\_TV & (f) Two step: TV\_OLS
\end{tabular}
\vspace{1cm}
\caption{Coefficient Maps for Simulation Setting 2 in 2D.}
\label{coefexample2}
\end{figure}

We consider two different settings in the simulation. 

In the first setting, we simulate 3 predictors as indicators of 3 disease groups and an integer value predictor as age. We generated 3 patient-level variables: $X_1$ and $X_2$ are disease group dummy variables, $X_{1}=0, X_{2}=0$  denote  the control group, $X_{1}=1, X_{2}=0$ represent the disease group 1, and $X_{1}=0, X_{2}=1$ represent the  disease group 2. A integer value variable $X_3$ is generated uniformly over integers from $56$ to $75$.
The image size is $m_1=40$ by $m_2=40$,  the true coefficient maps ($\Gamma_1,\Gamma_2,\Gamma_3$) are block-wise constant (as shown in top left panel of Figures~\ref{coefexample1}). And the outcome is generated by 
\bas{
Y = X_1\Gamma_1 + X_2\Gamma_2 + X_3\Gamma_3+\epsilon
}
where $\epsilon\in \bR^{m_1\times m_2}$, and $\epsilon_{ij} \stackrel{i.i.d.}{\sim} \cN(0,2)$.

In the second setting, two binary variables $X_1$ and $X_2$ are generated same as in setting 1. The coefficient map of $X_1$ is generated to have active regions with different sizes: $1$, $4$,  $25$ pixels, and the true coefficients are $2$, $1.5$ and $1$, respectively. Please refer to Figure \ref{coefexample2} for the actual arrangements of these block-wise true signals.

In Figure~\ref{coefexample1} and~\ref{coefexample2}, we present one example of the estimation for various methods  with training sample size of 100. The FoSR is designed for one-dimensional functionals, so we implemented it with the input as the vectorization of the image, this is the main reason why it only maintains the continuity along horizontal directions and fails to detect other types of structure. There are no off-the-shelf implementation of 2D FoSR as far as we know. 

According to Table~\ref{table2}, with increasing sample sizes, all methods perform better, but our method outperforms the others in all scenarios. Our proposed method outperformed the other methods except for setting 1 with sample size 100.  Envelope method demonstrates better performance with the large sample size for setting 1 which also has larger 'block size'. Envelope method is outperformed by our proposed method and  \texttt{TV\_OLS} (regression after TV denoising) for setting 2 with all three sample sizes, when the block has various sizes. 

According to Figures~\ref{coefexample1} and~\ref{coefexample2} that our proposed method outperforms the other methods in recovering the true structure of the parameters, and achieves a cleaner cut on the boundary. Envelope method is able to capture and display the main pattern, however does not encourage sparsity of the edges, thus the solution is not ``smooth" anywhere. 
TV\_OLS (regression after TV denoising) estimated coefficient maps are more 'vague' at the boundaries, since the the smoothing step is done separately for individual images, the estimated change points/edges will not be consistent across the smoothed images.  For setting 2, TV\_OLS performs well when sample size is small, but it does not improve much when sample size increases, since the TV denoising for single images can only recover piecewise constant signals for the moderate-size blocks. The signal from small regions will be smoothed out regardless of increasing the sample size because the denoising step only involves data from one image. In comparison, our proposed method gains much performance improvement with increasing sample size, as the signal is strengthened by more observations. 
By simultaneously conducting Image-on-scalar Regression and Total Variation Smoothing, our proposed method distinguishes signal on small regions (which is shared in all observations) versus random noise (which has no consistent behavior across different images), while competing methods fail to do so.
OLS\_TV performed well for setting 1 and has a  ``smoother" estimation for the coefficient map than our method.  But it fails to detect the small regions in setting 2 and performed worse than TV\_OLS. 

To summarize, TVTR is a universal method for recovering piecewise constant signals in small and large areas. It shows evident advantages of recovering signal in small areas compared with the two-stage methods and envelope method.

\section{Application to Activity Data}
\label{sec:real}
As discussed in the introduction, for the NIMH family study of spectrum disorders, we are interested in removing noises from the minute-level physical activity measures, and meanwhile accounting for heterogeneity in physical activity due to differences in demographic characteristics as well as health-related factors such as disease groups. Hence the outcome measures in our example are one-dimensional time series of activity intensities assessed over two weeks for each participant. Given the skewness of the observed activity counts \citep{Shou2017}, we conducted a log (counts+1) transformation beforehand and pre-smoothed the outcome data using moving average filter over a $60$-minute window, that is, averaged across an hour. As shown in Figure \ref{raw}, the raw activity counts were fairly noisy, but we assume that the underlying activity status for a particular participant (e.g., being inactive, moderately or vigorously active) would transit smoothly and continuously over time. Therefore such data might benefit from total variation smoothing that would restrict dramatic fluctuations from minute to minute. The covariates that we accounted in the model include diagnosis, age, gender, body mass index (BMI) and day of week. The associations between time-varying physical activity intensities and each of these scalar variables were estimated using our proposed model. To make the results more clinically interpretable, we categorized BMI into underweight ($<18.5$), normal ($18.5\le$BMI$\le25$), overweight ($25<$BMI$\le30$) and obese (BMI$>30$), with the normal participants being the reference group. Similarly, age was also stratified into four categories: adolescence (under 18), adulthood (18 to 40), middle age (40 to 60) and elderly ($\ge 60$). There were 5 diagnostic groups, including healthy control, type I bipolar (BPI), type II bipolar (BPII), major depressive disorder (MDD), and other disorders. The sample included 339 subjects with 14 days of data and our final analysis was based on 301 subjects with non-missing covariates. In total, there were 4214 daily activity curves and the observations were balanced for the 7 days of the week (Monday through Sunday). 

\begin{figure}[h!]
 \centering
\begin{tabular}{cc}
\includegraphics[width=0.4\textwidth]{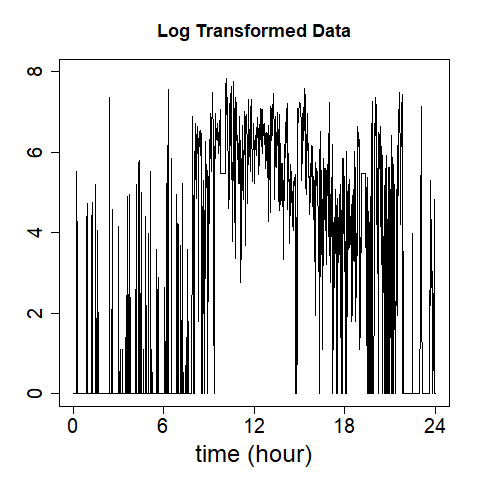}\hspace{-.6em}&
\includegraphics[width=0.4\textwidth]{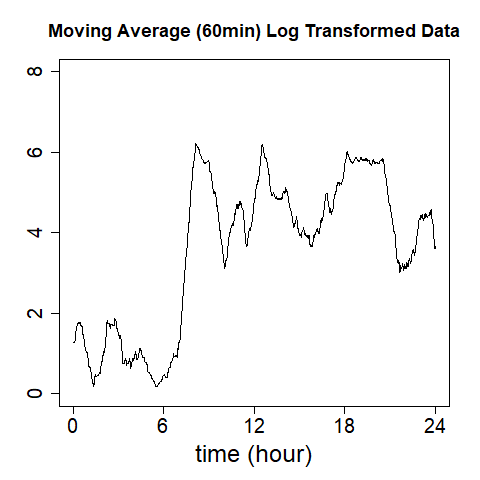} \hspace{-.6em}
\end{tabular}
\vspace{1cm}
\caption{An example of the observed daily activity counts after log transformation with (right) and without (left) pre-smoothing with moving average. }
\label{raw}
\end{figure}

The estimated time-varying effects for each covariate are shown as the red curve in Figure \ref{fig:asir1}. The point-wise confidence bands in red were obtained as the upper and lower $2.5\%$ quantiles over $100$ bootstrap samples. We observed that compared to the middle-age group, seniors tend to have an overall lower activity intensity throughout the day, especially from late afternoon to night. Both adolescent and adult groups have shown higher night-time activity and lower early morning activity, but the difference is more prominent between adolescent and middle-aged group. Participants who were obese were observed to have lower day-time activity and higher nighttime activity as compared to subjects with normal weights. The overweight subjects showed a similar trend but with a wider confidence band and smaller magnitude, so the trend were not significant. The underweight subjects were observed to have a reversed trend as compared to normal. In terms of diagnostic groups, BPI patients showed lower activity intensity later of the day (after 12pm) as compared with healthy controls, which is consistent with the findings reported in \citep{Shou2017}. Females were also observed to have higher average activity intensities during day time and lower intensities at night. Interestingly, when investigating the activity patterns over 7 days a week with Sunday as the reference, we independently observed similar patterns across weekdays (Monday to Friday), where the morning activity intensities are much higher than Sunday, but no difference were shown for the rest of the day. Note that on Friday, we observed a rising level of activity towards midnight, indicating that subjects were engaged in more late-night activities on Friday. Saturday demonstrated a non-differentiable pattern compared with Sunday. But a similar rising pattern were also observed towards midnight. 

To further evaluate our proposed method, we equally split the two weeks of activity data per subject into the training and testing datasets, where each dataset contained about 2107  days of measures respectively from the same set of subjects. We compared our proposed model with several competing methods including the minute-wise OLS regression, the function-on-scalar regression (FoSR) with B-spline basis functions, FoSR with $L_2$ penalty and the envelope method. The models were fit based on the training data and evaluated on the testing dataset by calculating the corresponding mean square errors (MSE) between the predicted activity counts and the observed counts. The prediction error (MSE) of our proposed method is 2.76, which is the smaller than both the minute-wise OLS (2.79), the FoSR with B-splines (2.80), as well as FoSR with $L_2$ penalty (2.79). The envelope method that had similar performance in the simulation studies demonstrated large MSE and is not suitable for the 1D application. As a reference, the total sample variance in the testing set is 2.90 which indicates that the activity data were very noisy. Although neither of the methods did a perfect job predicting activity counts, our proposed method were able to obtain interpretable estimates of for the association between activity and covariates and meanwhile demonstrate some improvement in prediction than the other methods.

\begin{figure}[H]
\centering
\vspace{-2.7em}
\caption{ Activity Data Analysis Results. } 
\label{fig:activity}
\begin{tabular}{cccc}
\hspace{-2em}
\includegraphics[width=0.25\textwidth]{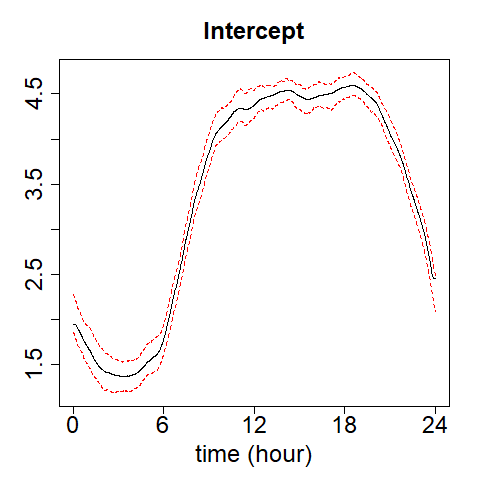}\hspace{-2em}&
\includegraphics[width=0.25\textwidth]{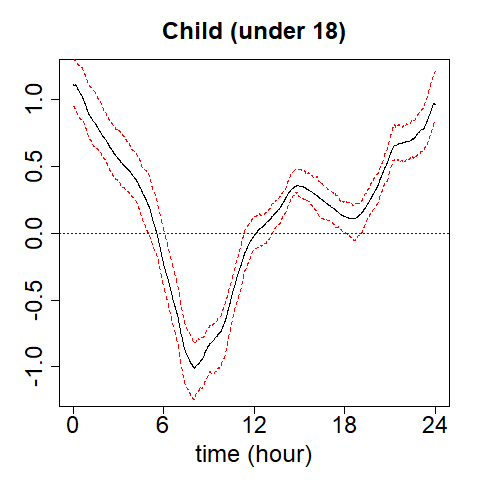} \hspace{-1em}&
\includegraphics[width=0.25\textwidth]{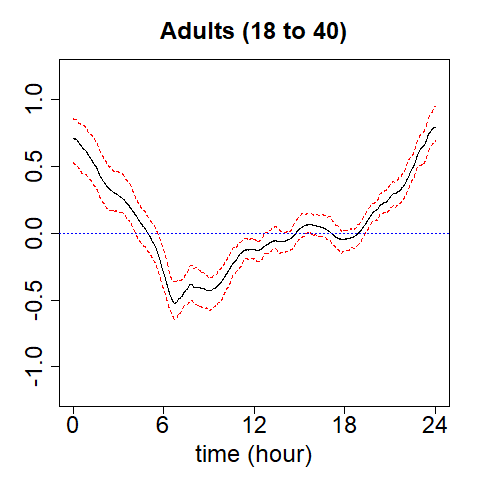} \hspace{-1em}&
\includegraphics[width=0.25\textwidth]{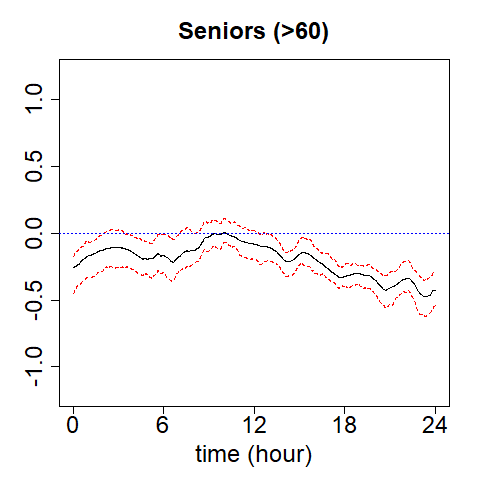} \hspace{-1em}\\
\hspace{-2em}
\vspace{-0.7em}
\includegraphics[width=0.25\textwidth]{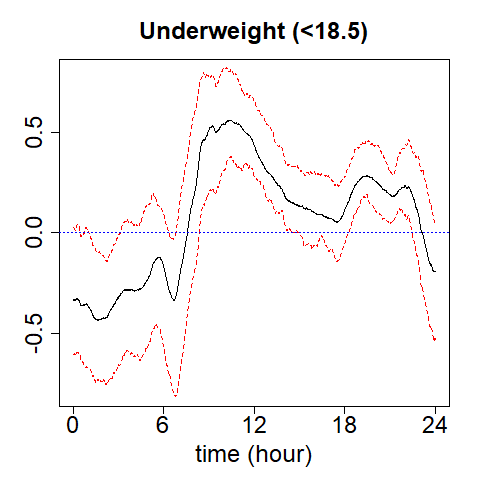}  \hspace{-1em}&
\includegraphics[width=0.25\textwidth]{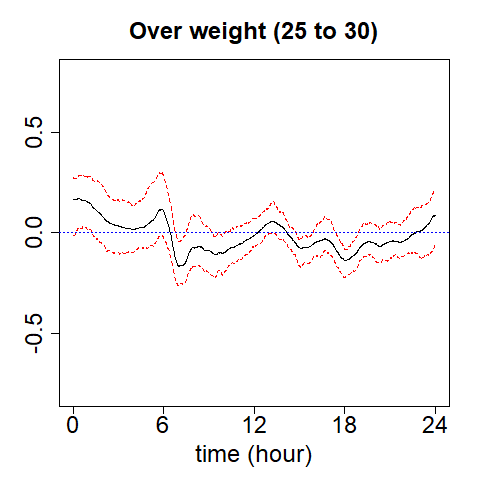} \hspace{-1em}&
\includegraphics[width=0.25\textwidth]{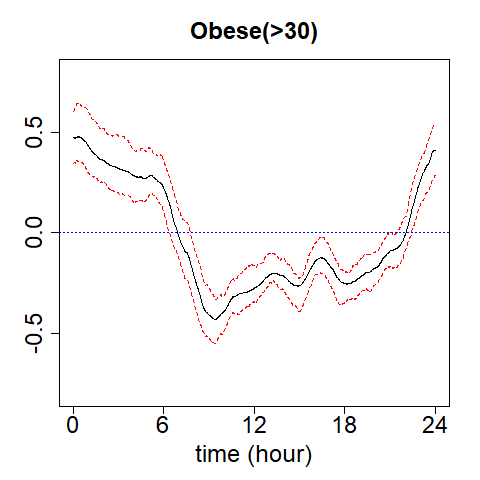}\hspace{-1em}&
\includegraphics[width=0.25\textwidth]{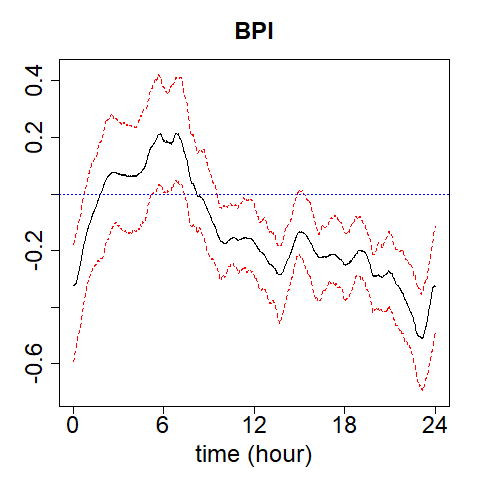}\hspace{-1em}\\
\hspace{-2em}\vspace{-0.7em}
\includegraphics[width=0.25\textwidth]{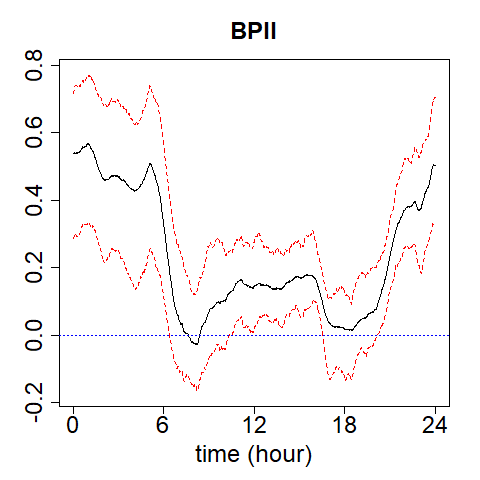}\hspace{-1em}&
\includegraphics[width=0.25\textwidth]{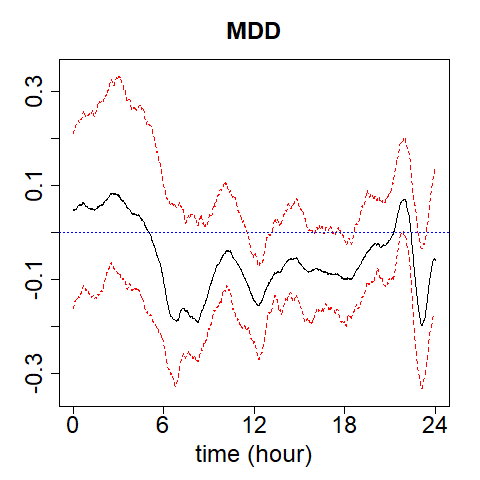}&
\includegraphics[width=0.25\textwidth]{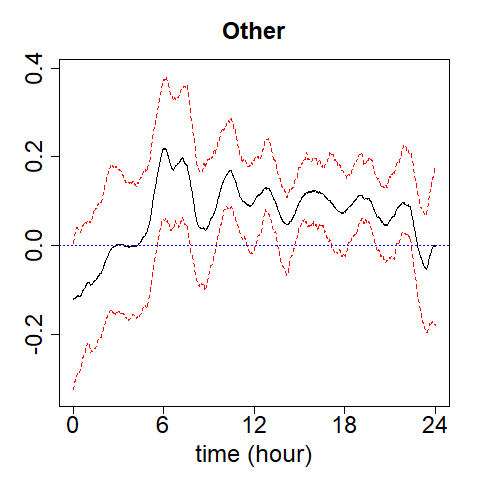}&
\includegraphics[width=0.25\textwidth]{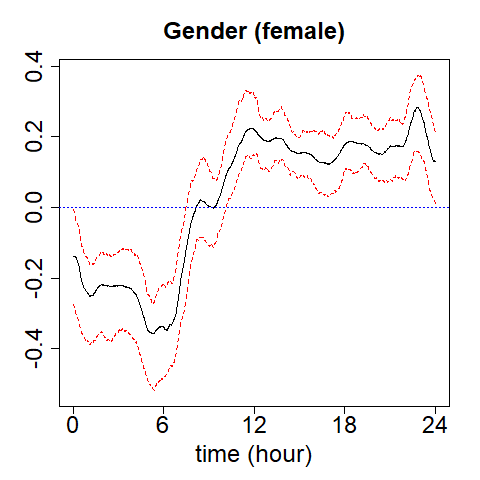}\\
\hspace{-2em}\vspace{-0.7em}
\includegraphics[width=0.25\textwidth]{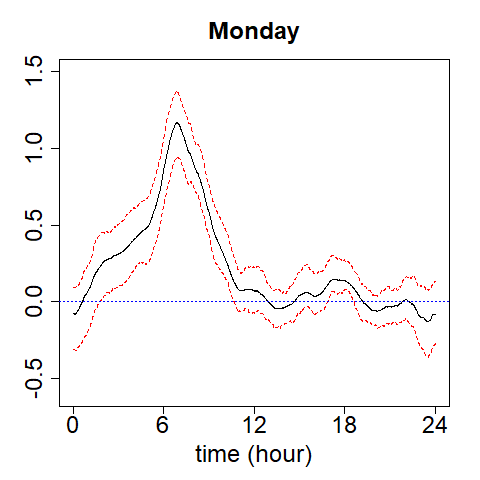}\hspace{-1em}&
\includegraphics[width=0.25\textwidth]{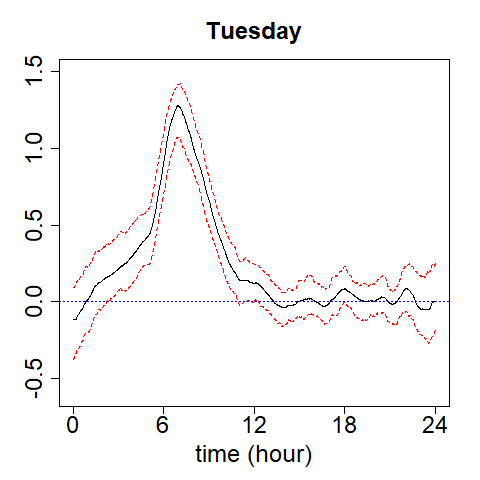}\hspace{-1em}&
\includegraphics[width=0.25\textwidth]{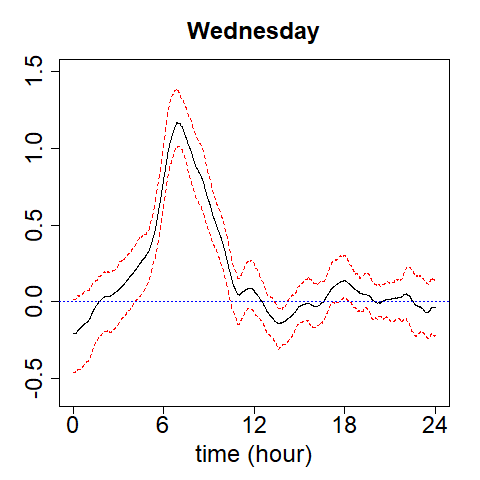}\hspace{-1em}&
\includegraphics[width=0.25\textwidth]{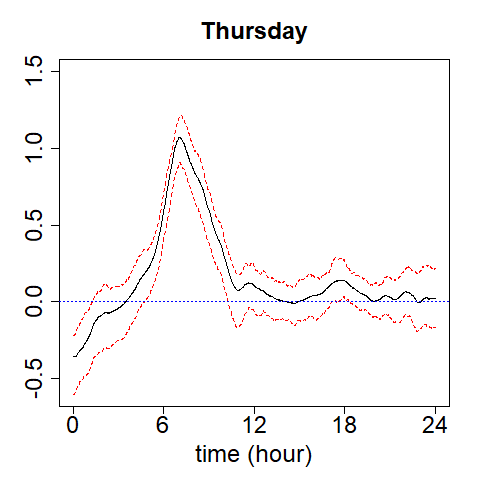}\hspace{-1em}\\
\hspace{-2em}\vspace{-0.7em}
\includegraphics[width=0.25\textwidth]{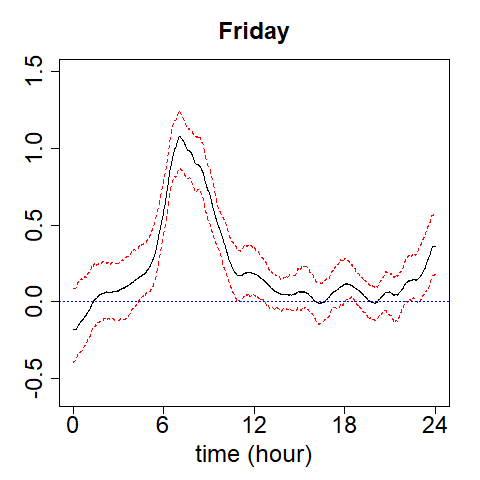}\hspace{-1em}&
\includegraphics[width=0.25\textwidth]{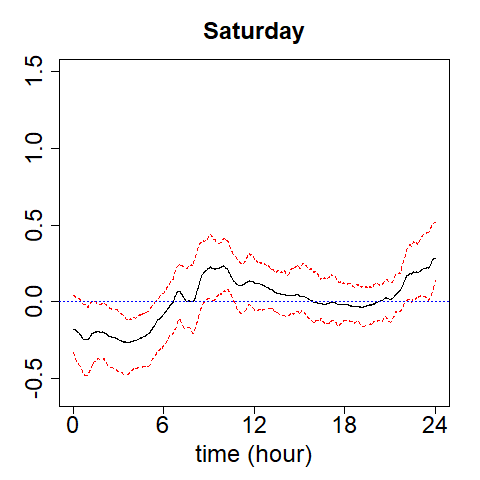}\hspace{-1em}\\
\end{tabular}
\label{fig:asir1}
\end{figure}

\section{Application to 3D Brain Imaging Analysis}
We now demonstrate the proposed method in analyzing 3D brain imaging. Our data came from a public data source of attention hyperactivity disorder (ADHD), the ADHD-200 sample, which was an initiative effort to understand neuro basis of the disease by ADHD consortium led by Dr. Michael Milham. The dataset contains structural and BOLD functional magnetic resonance imaging (MRI) scans from 776 subjects, including 491 typically developing controls (TDs) and 285 ADHD children. 

We focused on the gray matter voxel-based morphometry (VBM) maps preprocessed via the Burner pipeline using SPM8, downloaded from the Neuro Bureau (\url{https://www.nitrc.org/plugins/mwiki/index.php?title=neurobureau:BurnerPipeline}). The intensities of the images quantifies the normalized gray matter concentration changes at specific brain locations. The same dataset was used in the tensor envelope based regression method by \citep{li2017parsimonious}, which is one of the state-of-the-art methods for the 3D regression problem. Thus we compare the results of our proposed methods with the tensor envelope based regression and OLS estimates presented in \citep{li2017parsimonious}. The images are of dimension 121$\times$ 145$\times$ 121. Following \citep{li2017parsimonious}, we downsized the images to 30$\times$ 36$\times$ 30 using TensorReg. The covariates in the regression model include ADHD, age, gender, handedness. We also included the site indicators in the multivariate analysis to adjust for the systematic site effects. The estimated regression coefficients were mapped back to the 3D population-average template of the Burner VBM images. We also registered Harvard-Oxford Atlas \citep{Bohland2009} to the same template for references of the cortical regions. 

The computation cost of 3D data raised a challenge to many algorithms including ours. This specific example consists images from 770 subjects, each consists 32400 nodes and 94140 edges for a grid smooth structure.  Algorithm \ref{alg:palm} offers an iterative procedure of solving a TV denoising problem. In this real data example and our simulation studies, the algorithm converges in 10-20 iterations, when the parameters are initialized with random noise following normal distribution. The most computational expensive step of Algorithm \ref{alg:palm} is line 6, where one is solving a TV denoising problem for a graph with $32400\times 770$ nodes and $94140 \times 770$ edges. Solving TV denoising in a single step with this scale will cost memory failure on a personal computer. Fortunately, this task can be distributively solved. Since there are no edges across subjects, one can solve it by divide and conquer. i.e. within each iteration, at line 6,  one could divide the samples to several subgroups, solve for the $\mu^{(k+1)}$ for each subgroup, and combine the results together. The combined estimates of  $\mu^{(k+1)}$ is equivalent to the solutions from computing the combined TV denoising problem. For this dataset, we solved line 6 by `divide and conquer' with mini-batch size of 8 samples, each with the tuning parameter $\lambda=0.05$. The computation was conducted on a Windows 10 desktop computer with Intel(R) Core(TM) i7-7820U CPU@3.60GHz processor (16 cores), 32.00 GB installed memory(RAM).  It takes approximately 10 hours for obtaining the estimates under a chosen tuning parameter. Admittedly, the computational challenge caused problem in conducting cross validation for tuning parameters and for obtaining the bootstrap confidence intervals on a single machine. However, the algorithm is highly distributed and bootstrap would be practical using parallel computing resources.

Here we demonstrate the estimated regression coefficients in 3D representation based on our proposed method. The coefficient maps correspond to the association between voxel-wise gray matter morphometry with ADHD, age, gender and handedness (1 for left-handedness and 0 for right-handedness). We only demonstrated the areas where the coefficients' magnitudes are in the 10\%, 5\%, 2.5\% and 0.5\% tails across the whole brain. 

As shown in the Figure \ref{fig:ADHDcoef}, we observe that after adjusting for age, gender, handedness and site effects, ADHD showed the largest reduction in the gray matter volume in the frontal lobe such as the anterior and posterior cingulate gyrus and frontal operculum. Such regions have been known to be associated with sensory and motor responses, as well as cognitive and emotional functions. Significant cortical thinning in anterior cingulate gyrus has been reported among children with ADHD \citep{vanRooij2015}. These findings were novel in our analysis methods, while in the results of the envelope method \citep{li2017parsimonious} the effect of frontal lope were not captured. Instead, their methods identified superior temporal gyrus, pyramid and uvula in cerebellum. Additionally, we observed that cortical thickness was increased in frontal operculum with older children, which is linked with better task controls\citep{Higo2011}. Females are shown to have thicker temporal fusiform cortex and frontal medial cortex. Interestingly, although the magnitude of the handedness effects are small, we observed an asymmetric association where the left-handed children tend to have greater gray matter thickness on the left hemisphere of brain in regions such as posterior parahippocampal gyrus, lingual gyrus and temporal occipital fusiform cortex. Such findings are consistent with several previous studies \citep{cuzzocreo2009}.  

\begin{figure}[h!]
 \centering
\begin{subfigure}{0.45\textwidth}
\includegraphics[height=1.4in]{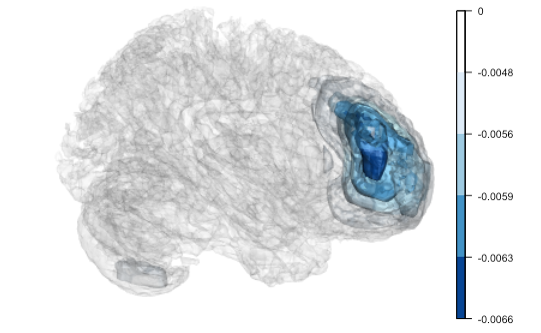}
\caption{ADHD}
\end{subfigure}
\begin{subfigure}{0.45\textwidth}
\includegraphics[height=1.4in]{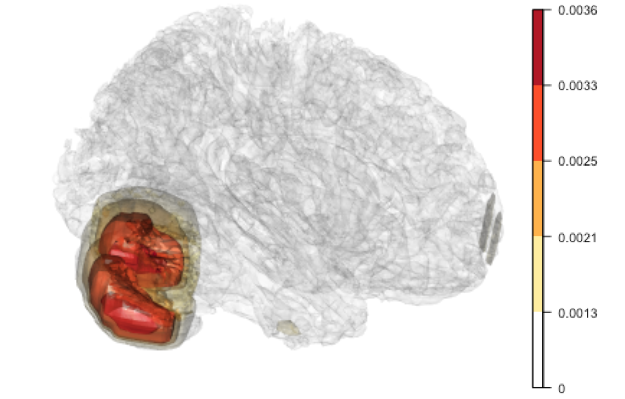}
\caption{Age}
\end{subfigure}
\begin{subfigure}{0.45\textwidth}
\hspace{.3em}
\includegraphics[height=1.45in]{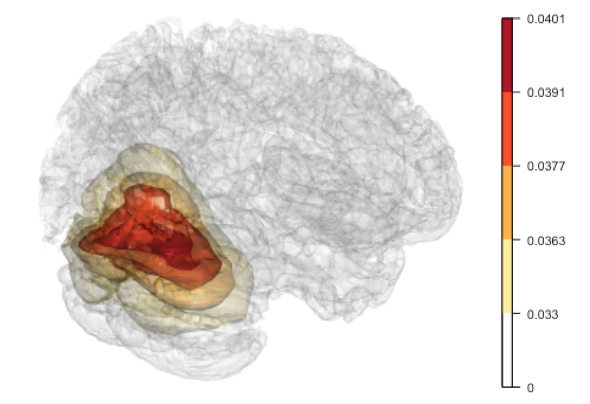}
\caption{Gender}
\end{subfigure}
\begin{subfigure}{0.45\textwidth}
\hspace{1em}
\includegraphics[height=1.5in]{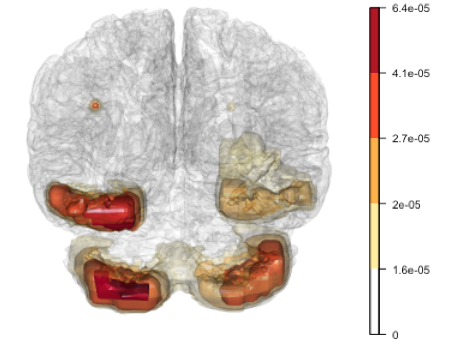} 
\caption{Left Handed}
\end{subfigure}

\caption{\small{ADHD analysis results: from light to dark are regions with estimated regression coefficients that lie above the 10\%, 5\%, 2.5\% and 0.5\% tails across the whole brain, receptively for covariates including ADHD, age, gender (females) and left-handedness. The color bars correspond to the positive or negative coefficients that are shown in the images.}}
 \label{fig:ADHDcoef}
\end{figure}

\section{Conclusion and Discussion}
\label{sec:discussion}

We propose a new method for tensor-on-scalar regression with a TV regularization term to encourage sparsity of adjacent outcome estimates. The proposed method  can be universally applied to tensor-on-scaler regression problem with any dimensions.  The proposed  algorithm is showed to converge, it is both scalable and consistent. Our method also has the appealing new feature of incorporating prior knowledge for the smoothness by user-defined adjacency matrix.  A extension of application of the method is in  brain connectivity studies, where the smoothness can be defined as the connected area. 

Although the method focuses on providing point estimation. In medical translational researches, in addition to signal recovering accuracy, proper inference of variability of the estimation is desirable. Thus we  demonstrate bootstrap confidence interval can be computed to help interpreting the result in the real data example of activity data analysis.  However, as dimension of the outcome increases from 1D to 3D, the computational burden increases tremendously that bootstrap CI is might be too costly to compute without utilizing powerful parallel computing cluster, which is a limitation of our current work.

Our current method only estimate the mean association not taking into account correlation of tensor-outcomes. It remains for future research how to more efficiently estimate the mean, with multiple outcome tensors per subject by adding random effect in the model.

\section*{Acknowlegements}
\label{sec:acknowledgement}

The authors would like to thank Drs. Xin Zhang and Lexin Li for providing codes and guidance of their envelope tensor regression algorithms. We also would like to acknowledge Junchi Li and Michael Daniels for discussion and advices to improve the quality of the work, and Jinshan Zeng for his comments on the convergence properties of the algorithm. HS's research was partially supported by the Intramural Research Program of the NIMH. The NIMH family study  was conducted under the support of by the Intramural Research Program (Merikangas; grant number Z-01-MH002804) and under clinical protocol 03-M-0211 (NCT00071786).

\section*{Disclaimer}
The views and opinions expressed in this article are those of the authors and should not be construed to represent the views of any of the sponsoring organizations, agencies or US Government.


\bigskip
\begin{center}
{\large\bf SUPPLEMENTARY MATERIAL}
\end{center}

{\small
\bibliographystyle{agsm}
\bibliography{fused}
}

\appendix

In this appendix we present technical details and accompanying lemmas which are necessary for the main results in the paper  {\it Total Variation Regularized Tensor-on-scalar Regression}. When we make references to equations or theorems etc. in the main document, we follow the numbering scheme of the main document, and the references do not have any alphabets in them.

\section{Proofs in Section~\ref{sec:model}}
\label{app:proof}
\subsection{Proof of Lemma~\ref{lem:theta2gamma}}
\begin{proof}[Proof of Lemma~\ref{lem:theta2gamma}]
Recall that $H_X$ is the projection matrix onto $\text{span}(X)^\perp$, therefore $\theta\in \text{span}(X)$ is equivalent to $H_X\theta=0$. By a variable transformation $\theta=X\Gamma$, solving (P)  is equivalent to solving (CP), and $\that=X\hat{\Gamma}$. When $X^TX$ is invertible, we solve $\hat{\Gamma}=(X^TX)^{-1}X^T \that$.
\end{proof}

\subsection{Proof of Theorem~\ref{thm:admm_converge}}
\begin{proof}[Proof of Theorem~\ref{thm:admm_converge}]
We rewrite problem~\eqref{eq:augmented_admm} as
\bas{
\min & \quad f(\theta)+g_1(\eta)+g_2(\mu)\\
s.t. & \quad A_1\theta+A_2\eta+A_3\mu=0
}
where
\bas{
&f(\theta)=\frac{1}{2}\|y-\theta\|^2, g_1(\eta) = \delta((I-H_v)\eta=0), g_2(\mu)=\|D_v^T\mu\|_1\\
& A_1 = \begin{bmatrix} -I\\-I\end{bmatrix}, A_2  = \begin{bmatrix} I\\0\end{bmatrix}, A_3  = \begin{bmatrix} 0 \\ I \end{bmatrix} \in \bR^{2nM\times nM}
}
Algorithm~\ref{alg:palm} is presented in the scaled form of ADMM. As is shown in~\citet{boyd2011distributed}, it is equivalent to the unscaled form as below.
Denote the Lagrangian as 
\bas{
\cL(\theta,\eta,\mu,U,V,\rho) = f(\theta) + g_1(\eta) + g_2(\mu) + \begin{bmatrix}U\\V\end{bmatrix}^T(A_1\theta+A_2\eta+A_3\mu)+\frac{\rho}{2}\|A_1\theta+A_2\eta+A_3\mu\|^2.
}
The unscaled form is
\bas{
& \theta^{(k+1)} = \arg\min_\theta  f(\theta) + \begin{bmatrix}U\\V\end{bmatrix}^T(A_1\theta)+\frac{\rho}{2}\|A_1\theta+A_2\eta^{(k)}+A_3\mu^{(k)}\|^2\\
& \eta^{(k+1)} = \arg\min_{\eta} g_1(\eta) + \begin{bmatrix}U\\V\end{bmatrix}^T(A_2\eta)+\frac{\rho}{2}\|A_1\theta^{(k+1)}+A_2\eta+A_3\mu^{(k)}\|^2\\
& \mu^{(k+1)} = \arg\min_{\mu} g_2(\mu) + \begin{bmatrix}U\\V\end{bmatrix}^T(A_3\mu)+\frac{\rho}{2}\|A_1\theta^{(k+1)}+A_2\eta^{(k+1)}+A_3\mu\|^2\\
& \begin{bmatrix}U^{(k+1)}\\V^{(k+1)}\end{bmatrix} = \begin{bmatrix}U^{(k)} \\V^{(k)} \end{bmatrix}+\rho(A_1\theta^{(k+1)}+A_2\eta^{(k+1)}+A_3\mu)
}
We first convert the three-block problem into a two-block problem with the technique in~\cite{chen2016direct}.

By first-order optimality condition in iteration $k$, we have
\bas{
& f(\theta)-f(\theta^{(k+1)}) + (\theta-\theta^{(k+1)})\left( -A_1^T(\lambda^k-\rho(A_1\theta+A_2\eta+A_3\mu)) \right)\ge 0, \forall \theta\in \bR^{nM}\\
& g_1(\eta) - g_1(\eta^{(k+1)}) + (\eta-\eta^{(k+1)})\left( -A_2^T(\lambda^k-\rho(A_1\theta+A_2\eta+A_3\mu)) \right)\ge 0, \forall \eta \in \bR^{nM}\\
& g_2(\mu)- g_2(\mu^{(k+1)}) + (\mu - \mu^{(k+1)})\left( -A_3^T(\lambda^k-\rho(A_1\theta+A_2\eta+A_3\mu)) \right)\ge 0, \forall \mu \in \bR^{nM}\\
}
Note that $A_2^TA_3 = 0$, we have
\bas{
& f(\theta)-f(\theta^{(k+1)}) + (\theta-\theta^{(k+1)})\left( -A_1^T(\lambda^k-\rho(A_1\theta+A_2\eta+A_3\mu)) \right)\ge 0, \forall \theta\in \bR^{nM}\\
& g_1(\eta) - g_1(\eta^{(k+1)}) + (\eta-\eta^{(k+1)})\left( -A_2^T(\lambda^k-\rho(A_1\theta+A_2\eta)) \right)\ge 0, \forall \eta \in \bR^{nM}\\
& g_2(\mu)- g_2(\mu^{(k+1)}) + (\mu - \mu^{(k+1)})\left( -A_3^T(\lambda^k-\rho(A_1\theta+A_3\mu)) \right)\ge 0, \forall \mu \in \bR^{nM}\\
}
which is also the first order optimality condition for the regime:
\bas{
&\theta^{(k+1)} = \arg\min f(\theta) - \lambda^{(k)T}(A_1\theta) + \frac{\rho}{2}\|A_1\theta+A_2\eta+A_3\mu\|^2;\\
& (\eta^{(k+1)},\mu^{(k+1)}) = \arg\min_{\eta,\mu} g_1(\eta)+g_2(\mu) - \lambda^{(k)T}(A_2\eta-A_3\mu) + \frac{\rho}{2}\|A_1\theta^{(k+1)}+A_2\eta+A_3\mu\|^2;\\
& \begin{bmatrix}U^{(k+1)}\\V^{(k+1)}\end{bmatrix} = \begin{bmatrix}U^{(k)}\\V^{(k)}\end{bmatrix} -\rho(A_1\theta^{(k+1)}+A_2\eta^{(k+1)}+A_3\mu^{(k+1)})
}
Clearly, this is a specific application of the two-block ADMM by regarding $(\eta, \mu)$ as one variable, $B:=[A_2, A_3]$ as one matrix, and $g(\eta,\mu):=g_1(\eta)+g_2(\mu)$ as one function. 
Existing convergence results for the two-block ADMM thus hold for our case.

Now we note that $f$ is Lipschitz differentiable and strongly convex, both $f(\theta)$ and $g(\eta,\mu)$ are closed (their sublevel sets are closed) and proper (they neither take on the value $-\infty$ nor are they uniformly equal to $\infty$)).
Also $A$ and $B$ both have full column rank. Therefore by Theorem 7 in~\citet{nishihara2015general}, Algorithm~\ref{alg:palm} converges to the unique global solution of \eqref{eq:obj} linearly.
\end{proof}

\section{Proofs in Section~\ref{sec:consistency}}
\label{sec:proof_consistent}
\subsection{Proof of Theorem~\ref{th:consist}}
\begin{proof}[Proof of Theorem~\ref{th:consist}]
Define $H_X=I-H_v$ is the projection matrix projecting each voxel to $\text{span}(X)$.
Let $\hat{\theta}\in \bR^{nM}$ be the optimal solution for the following constrained optimization problem.
\bas{
\min_\theta \quad & \| \ttvec(Y^T)-\theta \|_F^2+\lambda \|D_v^T\theta\|_{\ell_1}\\
\mbox{s.t.} \quad & H_v\theta=0.
}
When $\text{rank}(X)=p$, by Lemma~\ref{lem:theta2gamma}, $\hat{\Gamma} = (X^TX)^{-1}X^T\text{mat}(\theta)_{n\times M}$. We first prove an oracle inequality for $\hat{\theta}$.
By the KKT condition, $\exists z\in \sign(D_v^T\hat{\theta}), \alpha\in \bR^{nM}$ such that 
\ba{
& 2(\hat{\theta}-y)+\lambda D_v^Tz+H_v\alpha=0 \label{eq:first_order}\\
& H_v\hat{\theta}=0 \nonumber
}
Multiplying $H_X$ on the left of Eq.~\eqref{eq:first_order}, we have $2H_X(\hat{\theta}-y)+\lambda H_XD_v^Tz =0$.
Equivalently, 
\bas{\forall \bar{\theta} \in R^{nM},\qquad 2\bar{\theta}^T H_X(\hat{\theta}-y)+\lambda \bar{\theta}^T H_XD_v^Tz =0}
Note that $H_X$ is the projection matrix to the column space of $\text{span}(X)$, it suffices to consider $\forall \bar{\theta}=\ttvec(\bar{\Gamma}^TX^T)$.
By definition of the sign operator, the following holds:
\bas{
&\hat{\theta}^T (\ttvec(Y^T)-\hat{\theta}) = \lambda\|D_v^T\hat{\theta}\|_{\ell_1}\\
&\bar{\theta}^T (\ttvec(Y^T)-\hat{\theta}) \le \lambda\|D_v^T\bar{\theta}\|_{\ell_1}
}
Subtracting the former from the latter, and replacing $\ttvec(Y^T)$ with $\theta^*+\epsilon$, we get
\bas{
(\bar{\theta}-\hat{\theta})^T (\theta^*-\hat{\theta})\le (\hat{\theta}-\bar{\theta})^T  \epsilon +\lambda\|D_v^T \bar{\theta}\|_{\ell_1}-\lambda\|D_v^T\hat{\theta}\|_{\ell_1}
}
Note $(\bar{\theta}-\hat{\theta})^T (\theta^*-\hat{\theta}) = \frac{1}{4}\left( \|\bar{\theta}-\hat{\theta}\|^2+\|\theta^*-\hat{\theta}\|^2 - \|\bar{\theta}-\theta^*\|^2 \right)$,
\ba{
\|\bar{\theta}-\hat{\theta}\|^2+\|\theta^*-\hat{\theta}\|^2 \le \|\bar{\theta}-\theta^*\|^2+ 4(\hat{\theta}-\bar{\theta})^T  \epsilon +4\lambda\|D_v^T\bar{\theta}\|_{\ell_1}-4\lambda\|D_v^T\hat{\theta}\|_{\ell_1}
\label{eq:diff_raw}
}

To bound $(\hat{\theta}-\bar{\theta})^T  \epsilon $, note $DD^T$ is the graph Laplacian, therefore when the graph is connected, we have $\text{ker}(D^T)=\text{ker}(DD^T)=\ttspan\{1_M\}$. Define $D^\dagger$ be the pseudo inverse of $D$, then $I-(D^\dagger)^T D^T$ is the projection matrix onto ker($D^T$),
\bas{
(\hat{\theta}-\bar{\theta})^T \epsilon = & \sum_{i=1}^n (\hat{\theta}_i -\bar{\theta}_i)^T \epsilon_i\\
= & \sum_{i=1}^n   ((I-(D^\dagger)^T D^T)\epsilon_i)^T(\hat{\theta}_i-\bar{\theta}_i)+((D^\dagger)^T D^T\epsilon_i)^T(\hat{\theta}_i-\bar{\theta}_i)\\
\le &  \sum_{i=1}^n \|(I-(D^\dagger)^T D^T)\epsilon_i\|\cdot \|\hat{\theta}_i -\bar{\theta}_i\|+\|(D^\dagger)^T\epsilon_i\|_\infty \cdot \|D^T(\hat{\theta}_i-\bar{\theta}_i)\|_{\ell_1}.
}
In view of the fact that $\epsilon_i\sim \cN(0,\sigma^2I_M)$ and $(I-(D^\dagger)^T D^T)$ being a projection matrix to a one dimensional space, by the tail bound for Gaussian random variables, $\forall i\in [n], \forall \delta>0$, 
\bas{
P(\|(I-(D^\dagger)^T D^T)\epsilon_i\|\ge 2\sigma\sqrt{2\log(2enM/\delta)} )\le \frac{\delta}{2n}
}
For the second part, by the maximal inequality for Gaussian random variables (\cite{massart2007concentration} Thm 3.12), and the variance of elements of $(D^\dagger)^T)\epsilon_i$ is upper bounded by $\rho^2\sigma^2$,
\bas{
P(\|(D^\dagger)^T\epsilon_i\|_\infty \ge \rho\sigma\sqrt{2\log(2emnM/\delta)} )\le \frac{\delta}{2n}
}
Applying union bound, with probability at least $1-\delta$,
\beq{
\bsplt{
(\hat{\theta}-\bar{\theta})^T \epsilon \le& \sum_{i=1}^n \left( 2\sigma\sqrt{2\log(2enM/\delta)}\|\hat{\theta}_i -\bar{\theta}_i\|+ \rho\sigma\sqrt{2\log(2emnM/\delta)}\|D^T(\hat{\theta}_i-\bar{\theta}_i)\|_{\ell_1} \right)\\
\le &  2\sigma\sqrt{2\log(2enM/\delta)} \|\hat{\theta} -\bar{\theta}\|+ \rho\sigma\sqrt{2\log(2emnM/\delta)}\left( \sum_{i=1}^n \|D^T(\hat{\theta}_i-\bar{\theta}_i)\|_{\ell_1} \right)
}
\label{eq:theta_eps}
}

By the triangle inequality, 
\bas{
\|D^T(\that-\tbar)_{T^c}\|_{\ell_1}-\|D^T\that_{T^c}\|_{\ell_1}\le \|D^T\tbar_{T^c}\|_{\ell_1}\\
\|D^T\tbar_{T}\|_{\ell_1}-\|D^T\that_{T}\|_{\ell_1}\le \|D^T(\tbar-\that)_{T}\|_{\ell_1}
}
Hence
\bas{
&\|D^T(\hat{\theta}_i-\bar{\theta}_i)\|_{\ell_1}+\|D^T(\bar{\theta}_i)\|_{\ell_1}-\|D^T(\hat{\theta}_i)\|_{\ell_1} \\
=& \|D^T(\hat{\theta}_i-\bar{\theta}_i)_T\|_{\ell_1}+\|D^T(\hat{\theta}_i-\bar{\theta}_i)_{T^c}\|_{\ell_1}+\|D^T(\bar{\theta}_i)_T\|_{\ell_1}+\|D^T(\bar{\theta}_i)_{T^c}\|_{\ell_1}-\|D^T(\hat{\theta}_i)_T\|_{\ell_1}-\|D^T(\hat{\theta}_i)_{T^c}\|_{\ell_1} \\
\le& 2\|D^T(\hat{\theta}_i-\bar{\theta}_i)_T\|_{\ell_1}+2\|D^T(\bar{\theta}_i)_{T^c}\|_{\ell_1}
}
By Definition 1 ,
$
\|D^T(\hat{\theta}_i-\bar{\theta}_i)_T\|_{\ell_1} \le \kappa_T^{-1}\sqrt{|T|}\|\hat{\theta}_i-\bar{\theta}_i\|.
$
We now plug above and \eqref{eq:theta_eps} back to \eqref{eq:diff_raw}, and take $\lambda = \rho\sigma\sqrt{\log(mnM/\delta)}$, then with probability at least $1-c_7n^{-1}$,
\beq{
\bsplt{
\|\bar{\theta}-\hat{\theta}\|^2+\|\theta^*-\hat{\theta}\|^2 \le & \|\bar{\theta}-\theta^*\|^2+ 8\sigma\sqrt{2\log(2en/\delta)} \|\hat{\theta} -\bar{\theta}\| + 4\lambda \|(D\bar{\theta})_{T^c}\|_{\ell_1} + 4\lambda \kappa_T^{-1}\sqrt{|T|}\|\hat{\theta}-\bar{\theta}\|}
\label{eq:tbar_minus_that}
}
Use Young's inequality, 
\bas{
8\sigma\sqrt{2\log(2en/\delta)} \|\hat{\theta} -\bar{\theta}\|\le & \frac{1}{2}\|\that-\tbar\|^2+64\sigma^2\log\left( \frac{2en}{\delta} \right)\\
4\lambda \kappa_T^{-1}\sqrt{|T|}\|\hat{\theta}-\bar{\theta}\| \le & \frac{1}{2}\|\that-\tbar\|^2+ 8\lambda^2\kappa_T^{-2}|T|
}
Canceling out $\|\that-\tbar\|^2$ on both sides of \eqref{eq:tbar_minus_that}, 
\bas{
\|\theta^*-\hat{\theta}\|^2 \le & \|\bar{\theta}-\theta^*\|^2+ 4\lambda \|(D_v^T\bar{\theta})_{T^c}\|_{\ell_1} +  64\sigma^2\log\left( \frac{2enM}{\delta} \right)+ 8\lambda^2\kappa_T^{-2}|T|
} 

Taking infimum on the right and plugging in $\lambda$ we have 
\bas{
&\|\theta^*-\hat{\theta}\|^2\le \\
 & \inf_{\bar{\theta}\in \bR^{n M}: H_X\bar{\theta}=\bar{\theta}} \left\{ \|\bar{\theta}-\theta^*\|^2+ 4\lambda \|(D_v^T\bar{\theta})_{T^c}\|_{\ell_1} \right\} + 64\sigma^2\log\left( \frac{2enM}{\delta} \right)+ 8\rho^2\sigma^2\log\left(\frac{mnM}{\delta}\right)\kappa_T^{-2}|T|
}
\label{eq:bound_theta}
\end{proof}

\subsection{Proof of Corollary~\ref{cor:gamma_oracle} }
Corollary~\ref{cor:gamma_oracle} can be proved by combining Lemma~\ref{lem:theta2gamma} and Theorem~\ref{th:consist}.

\begin{proof}[Proof of Corollary~\ref{cor:gamma_oracle}]
When $\frac{1}{n}X^TX=I_p$, we will be able to control the error in $\Gamma$, notice that $\|\ttvec(A)\|_2^2 = \|A\|_F^2$, we have from \eqref{eq:bound_theta} that 
\bas{
& \|\hat{\Gamma}-\Gamma^*\|_F^2 = \tr((X^TX)^{-1}X^T\text{mat}((\theta^*-\that) (\theta^*-\that)^T)X(X^TX)^{-1}) \\
= & \frac{1}{n} \|\theta^*-\hat{\theta}\|^2\\
=& \inf_{\bar{\theta}\in \bR^{n M}: H_X\bar{\theta}=\bar{\theta}} \left( \frac{1}{n} \|\bar{\theta}-\theta^*\|^2+ \frac{4\lambda}{n} \|(D_v^T\bar{\theta})_{T^c}\|_{\ell_1}  \right) + \frac{1}{n}\left( 4\sigma\sqrt{2\log(2enM/\delta)}+2\lambda \kappa_T^{-1}\sqrt{|T|}\right)^2\\
\le &  \inf_{\bar{\Gamma}\in \bR^{p\times M}} \left( \|\bar{\Gamma}-\hat{\Gamma}\|_F^2+ \frac{4\lambda}{ n} \|(X\bar{\Gamma}D)_{T^c}\|_{\ell_1} \right) + \frac{1}{n }\left( 4\sigma\sqrt{2\log(2enM/\delta)}+2\lambda \kappa_T^{-1}\sqrt{|T|}\right)^2
}
\end{proof}

\end{document}